\documentclass[english, onecolumn]{IEEEtran} % For LaTeX2e
\usepackage{hyperref}
\usepackage{url}
\usepackage{graphicx}
\usepackage{amsfonts}
\usepackage{amsgen,amsmath,amstext,amsbsy,amsopn,amssymb,amsthm}
\usepackage{algorithm}
\usepackage{multirow}
\newtheorem{definition}{Definition}

\newtheorem{lemma}{Lemma}
\newtheorem{theorem}{Theorem}

\newtheorem{corollary}{Corollary}

\newcommand{\be}{\begin{equation}}
\newcommand{\ee}{\end{equation}}

\newcommand{\bbR}{\mathbb{R}}
\newcommand{\cO}{\mathcal{O}}

\newcommand{\cL}{\mathcal{L}}

\newcommand{\0}{{\mathbf{0}}}

\renewcommand{\b}{{\mathbf{b}}}
\renewcommand{\c}{{\mathbf{c}}}
\newcommand{\e}{{\mathbf{e}}}

\newcommand{\g}{{\mathbf{g}}}
\newcommand{\h}{{\mathbf{h}}}
\renewcommand{\u}{{\mathbf{u}}}
\renewcommand{\v}{{\mathbf{v}}}
\renewcommand{\r}{{\mathbf{r}}}
\newcommand{\x}{{\mathbf{x}}}
\newcommand{\y}{{\mathbf{y}}}
\newcommand{\z}{{\mathbf{z}}}
\newcommand{\A}{{\mathbf{A}}}
\newcommand{\B}{{\mathbf{B}}}

\newcommand{\I}{{\mathbf{I}}}

\newcommand{\supp}{{\rm supp}}

\newcommand{\w}{{\mathbf{w}}}

\newcommand{\cN}{{\cal N}}

\usepackage{algorithmic}
\floatstyle{plain}
\newfloat{myalgo}{tbhp}{mya}

{\begin{myalgo}[#1]
\centering
\begin{minipage}{#2}
\begin{algorithm}[H]}%
{\end{algorithm}
\end{minipage}
\end{myalgo}}

\title{Alternating direction algorithms for $\ell_0$ regularization in compressed sensing}
\author{Chaobing Song, Shu-Tao Xia}

\begin{document}
\maketitle

\begin{abstract}
In this paper we propose three iterative greedy algorithms for compressed sensing, called \emph{iterative alternating direction} (IAD), \emph{normalized iterative alternating direction} (NIAD) and  \emph{alternating direction pursuit} (ADP), which stem from the iteration steps of alternating direction method of multiplier (ADMM) for $\ell_0$-regularized least squares ($\ell_0$-LS) and can be considered as the alternating direction versions of the well-known iterative hard thresholding (IHT), normalized iterative hard thresholding (NIHT) and hard thresholding pursuit (HTP) respectively. Firstly, relative to the general iteration steps of ADMM, the proposed algorithms have no splitting or dual variables in iterations and thus the dependence of the current approximation on past iterations is direct. Secondly, provable theoretical guarantees are provided in terms of restricted isometry  property, which  is the first theoretical guarantee of  ADMM for $\ell_0$-LS to the best of our knowledge. Finally, they outperform the corresponding IHT, NIHT and HTP greatly when  reconstructing both constant amplitude signals with random signs (CARS signals) and Gaussian signals. %The form of ADMM for $\ell_0$-LS is reformulated and has a close connection with the  gradient descent methods, such as IHT, NIHT and HTP.
\end{abstract}
%  \begin{keywords}
% Compressed sensing, restricted isometry constants, sparse recovery
%  \end{keywords}
\renewcommand{\thefootnote}{\fnsymbol{footnote}} \footnotetext[0]{
This research is supported in part by the Major State Basic Research
Development Program of China (973 Program, 2012CB315803), the National
Natural Science Foundation of China (61371078), and the Research Fund for the Doctoral Program of Higher Education of China (20130002110051)
.

All the authors are with the Graduate School at ShenZhen, Tsinghua University, Shenzhen, Guangdong 518055, P.R. China (e-mail: wordsword15@icloud.com,  xiast@sz.tsinghua.edu.cn).} \renewcommand{\thefootnote}{\arabic{footnote}} \setcounter{footnote}{0}

\section{Introduction}
As a new paradigm for signal sampling, compressed sensing (CS) \cite{donoho2006compressed,candes2005decoding,candes2006robust}
has attracted a lot of attention in recent years. Consider an $s$-sparse signal
$\x=(x_1,x_2,\ldots,x_n)^T\in \bbR^n$ which has at most $s$ nonzero entries. Let
$\A\in\bbR^{m\times n}$ be a measurement matrix with $m\ll n$ and $\b=\A\x$ be
a measurement vector. CS deals with recovering the original signal $\x$ from
the measurement vector $\b$ by finding the sparsest solution to the underdetermined
linear system $\b=\A\x$, i.e., solving the following \emph{$\ell_0$ minimization}
problem:
\be
\min \|\x\|_0\qquad s.t.\qquad \A\x=\b, \label{eq:l0}
\ee
where $\|\x\|_0:=|\{i:x_i\neq 0\}|$ denotes the $\ell_0$ quasi norm of $\x$.
Unfortunately, as a typical combinatorial optimization problem, the above $\ell_0$ minimization
is NP-hard \cite{candes2005decoding}.

One popular strategy is to relax the $\ell_0$ minimization
problem to an \emph{$\ell_1$ minimization} problem:
\be
\min \|\x\|_1\qquad s.t.\qquad \A\x=\b, \label{eq:l1}
\ee
which is a constrained linear programming and thus can be solved in polynomial   time \cite{nesterov1994interior} when interior point methods are employed. However, the complexity $O(m^2 n^{3/2})$  \cite{nesterov1994interior} of such a second-order method can be prohibitively expensive for the very large-scale problems that arise from typical CS applications, e.g., $n\approx 10^6$. Accordingly, recent years have witnessed a renewed interest in simpler first-order methods, which aim at solving an unconstrained problem called ``$\ell_1$-regularized least-squares ($\ell_1$-LS)'' as follows:
\be
\min_{\x} \|\x\|_1+\frac{1}{2\alpha}\|\A\x-\b\|_2^2, \label{eq:l1ls}
\ee
for an appropriately chosen $\alpha>0$ depending on the noise level. When there is no noise on the measurements b, or when the noise level is low, one can
solve \eqref{eq:l1ls} for a fixed small $\alpha>0$, which can be view as a penalty approximation to \eqref{eq:l1}. A lot of first-order algorithms were proposed to get an $\epsilon$-accuracy estimation, such as iterative shrinkage thresholding algorithm (ISTA) \cite{chambolle1998nonlinear}, fast iterative shrinkage thresholding algorithm (FISTA) \cite{beck2009fast}, gradient projection for sparse
reconstruction (GPSR) \cite{figueiredo2007gradient},  fixed point continuation (FPC) \cite{hale2008fixed} and so on.

 In \cite{yang2011alternating}, instead of solving \eqref{eq:l1ls} directly, the authors proposed alternating direction algorithms which employ the well-known alternating direction method of multiplier (ADMM) \cite{boyd2011distributed} to solve   \eqref{eq:l1ls} by variable splitting as follows:
\be
\min_{\x\in\bbR^{n},\r\in\bbR^{m}}\left\{\|\x\|_1+\frac{1}{2\alpha}\|\r\|_2^2: \A\x+\r=\b\right\}, \label{eq:admml1}
\ee
where $\r\in\bbR^{m}$ is a splitting variable.
Corresponding to the $\ell_1$-LS solvers ISTA and its variants, there exist iterative hard thresholding (IHT) \cite{Blumensath2009} and its variants such as  normalized iterative thresholding pursuit (NIHT) \cite{blumensath2010normalized}, hard thresholding pursuit (HTP) \cite{foucart2011hard}
 to solve the following ``$\ell_0$-regularized least-squares ($\ell_0$-LS)'' problem
 \be
  \min_{\x} \|\x\|_0+\frac{1}{2\alpha}\|\A\x-\b\|_2^2 \label{eq:l0ls}
\ee
directly without convex relaxation (In IHT, NIHT and HTP, $\alpha$ is changed implicitly to maintain $s$ nonzero entries in the approximation vector in each iteration.). These iterative greedy algorithms have comparable theoretical guarantees with $\ell_1$ minimization in terms of \emph{restricted isometry property} (RIP), good empirical performance and low computational complexity   \cite{blumensath2010normalized},   \cite{foucart2011hard}, \cite{maleki2010optimally}. Accordingly, instead of the $\ell_1$-regularization in \eqref{eq:admml1}, one may try to solve
 \be
\min_{\x\in\bbR^{n},\r\in\bbR^{m}}\left\{\|\x\|_0+\frac{1}{2\alpha}\|\r\|_2^2: \A\x+\r=\b\right\} \label{eq:admml0}
\ee
directly too. In \cite{boyd2011distributed}, the authors gave a beneficial discussion about this idea, but they did not give any theoretical guarantee and gave no connection with the existing iterative greedy algorithms.

In this paper, we give a further study about \eqref{eq:admml0}. Firstly, we reformulate the classical iteration steps of ADMM iteration to a new form without splitting and dual variables which is called the \emph{iterative alternating direction} (IAD) algorithm, and thus show its close connection with the well-known IHT algorithm. Then, two variants of IAD called \emph{normalized iterative alternating direction} (NIAD) and \emph{alternating direction pursuit} (ADP) are proposed which correspond to the variants NIHT and HTP of IHT. Moreover, the theoretical guarantees are given in terms of \emph{restricted isometry property} (RIP) for IAD, NIAD and ADP. Finally, experiments are given to show the improved empirical performance of  IAD, NIAD and ADP relative to the corresponding IHT, NIHT and HTP algorithms.

Notations: Let $\mathbf{x}\in\mathbb{R}^n$. Let $S\subseteq \{1,2,\ldots,n\}$, and $|S|$ and $\overline{S}$ respectively denote the cardinality and complement of $S$. Let $\mathbf{x}_S\in\mathbb{R}^n$ denote the vector obtained from $\mathbf{x}$ by keeping the $|S|$ entries in $S$ and setting all other entries to zero.  Let $\text{supp}(\mathbf{x})$ denote the support of $\mathbf{x}$ or the set of indices of nonzero entries in $\mathbf{x}$. Note that $\mathbf{x}$ is $s$-sparse if and only if $|\text{supp}(\mathbf{x})|\le s$. For a matrix $\mathbf{\A}\in\mathbb{R}^{m\times n}$, let $\mathbf{\A}^{\! T}$ denote the  transpose of $\mathbf{\A}$ and $\mathbf{\A}_S$ denote the submatrix that consists of columns of $\mathbf{\A}$ with indices in $S$. Let $\mathbf{I}$ denote the identity matrix whose dimension is decided by contexts. In addition, define $
\rm hard(\x,\tau)(\tau\ge0)$ the vector that set all the entries of $\x$ except the entries (in magnitude) larger than $\tau$ to zero and $H_s(\x)$ the vector that set all the entries of $\x$ except the $s$ largest magnitude entries to zero. Finally, for all series $\{c(k)\},k\in\{0,1,2\cdots\}$, we denote $\sum_{k=i}^{j}c(i)=0$, if $i>j$.

Denote the general CS model:
\begin{equation}
\mathbf{b}=\mathbf{\A}\mathbf{x}+\mathbf{e}=\mathbf{\A}\mathbf{x}_S+\mathbf{\A}\mathbf{x}_{\overline{S}}+\mathbf{e}=\mathbf{\A}\mathbf{x}_S+\mathbf{e}^{\prime},
\label{eq:general_model}
\end{equation}
where $\mathbf{\A}\in\mathbb{R}^{m\times n}$ is a measurement matrix with $m\ll n$, $\mathbf{e}\in\mathbb{R}^m$
is an arbitrary noise, $\mathbf{b}\in\mathbb{R}^m$ is a low-dimensional observation, $S$ denotes the indices of the $s$ largest magnitude entries of $\x$, and $\mathbf{e}^{\prime}=\mathbf{\A}\mathbf{x}_{\overline{S}}+\mathbf{e}$
denotes the total perturbation by the sparsity defect $\mathbf{x}_{\overline{S}}$ and measurement error $\mathbf{e}$.

\section{Reformulation of ADMM iteration}

\subsection{Applying ADMM to $\ell_0$-LS after variable splitting}
Consider the $\ell_0$-LS problem \eqref{eq:admml0} after variable splitting. The  augmented Lagrangian function of \eqref{eq:admml0} is give as follows
\be
\cL(\x,\r,\y)=\{\|\x\|_0+\frac{1}{2\alpha}\|\r\|_2^2-\y^{T}(\A\x+\r-\b)+\frac{\beta}{2}\|\A\x+\r-\b\|_2^2\}, \label{eq:2}
\ee
where the dual variable $\y\in\bbR^{m}$ is a multiplier and $\beta>0$ is a penalty parameter.
In \eqref{eq:2}, by utilizing the separability structure of $\x$ and $\r$ in the objective function, ADMM minimize $\cL(\x,\r,\y)$ with respect to $\x$ and $\r$ separately via a Gauss-Seidel type iteration. After minimizing $\r$ and $\x$ once in order, the multiplier $\y$ is updated immediately. The iteration steps can be explained as follows.

\begin{algorithm}[H]
\label{alg:general-admm}
Initialize $\x(0)$ and $\y(0)=0$, $k=0$.\\
Iteration: At the $k$-th iteration, go through the following steps.
\begin{enumerate}
\item $\r(k+1)=\arg\min_{\r} \cL(\x(k),\r,\y(k));$
\item $\x(k+1)=\arg\min_{\x} \cL(\x,\r(k+1),\y(k));$
\item  $\y(k+1)=\y(k)-\beta(\A\x(k+1)+\r(k+1)-\b).$
\end{enumerate}
until the stopping criterion is met. \\
Output: $\r(k+1),\x(k+1),\y(k+1)$.
\caption{ADMM}\label{alg:1}
\end{algorithm}

Firstly, in step 1 of  Alg. \ref{alg:1}, the minimizer of \eqref{eq:2} with respect to $\r$ is given by
\be
\r(k+1)=\frac{\alpha\beta}{1+\alpha\beta}\left(\y(k)/\beta+\b-\A\x(k)\right). \label{eq:3}
\ee
However, after some simple transform, the subproblem in step 2 of Alg. \ref{alg:1} is equivalent to
\be
\min_{\x\in\bbR^n}\|\x\|_0+\frac{\beta}{2}\|\A\x+\r(k+1)-\b-\y(k)/\beta\|_2^2,\label{eq:33}
\ee
which is the form of \eqref{eq:l0ls}. We approximately  solve \eqref{eq:33} by using a quadratic approximation of $\frac{1}{2}\|\A\x+\r(k+1)-\b-\y(k)/\beta\|_2^2$ at $\x=\x(k)$,
but keeping $\|\x\|_0$ intact:
\be
\min_{\x\in\bbR^{n}}\|\x\|_0+\beta\left(\g(k)^{T}(\x-\x(k))+\frac{1}{2\tau}\|\x-\x(k)\|_2^2\right), \label{eq:l01}
\ee
where $\tau>0$ is a proximal parameter and
\begin{eqnarray*}
\g(k)&=&\A^{\! T}\left(\A\x(k)+\r(k+1)-\b-\y(k)/\beta\right)\\
&=&-\frac{1}{\alpha\beta}\A^{\! T} \r(k+1)
\end{eqnarray*}
is the gradient vector of $\frac{1}{2}\|\A\x+\r(k+1)-\b-\y(k)/\beta\|_2^2$ at $\x=\x(k)$.
\eqref{eq:l01} can be solved explicitly by (see e.g., \cite{wright2009sparse})
\be
\x(k+1)=\rm hard\left(\x(k)-\tau\g(k),\sqrt{\frac{2\tau}{\beta}}\right)=hard\left(\x(k)+\frac{\tau}{\alpha\beta}\A^{\! T}\r(k+1),\sqrt{\frac{2\tau}{\beta}}\right).
\ee
Finally, the multiplier $\y$ is updated by
\be
 \y(k+1)=\y(k)-\beta\left(\A\x(k+1)+\r(k+1)-\b\right).
 \ee
 So, when applying ADMM to $\ell_0$-LS after variable splitting, the algorithm can be summarized as follows:
\begin{algorithm}[H]
Initialize $\x(0)$ and $\y(0)=0$, $k=0$.\\
Iteration: At the $k$-th iteration, go through the following steps.
\begin{enumerate}
\item $\r(k+1)=\frac{\alpha\beta}{1+\alpha\beta}\left(\y(k)/\beta+\b-\A\x(k)\right);$
\item $\x(k+1)=\rm{hard}\left(\x(k)+\frac{\tau}{\alpha\beta}\A^{\! T}\r(k+1),\sqrt{\frac{2\tau}{\beta}}\right);$
\item  $\y(k+1)=\y(k)-\beta(\A\x(k+1)+\r(k+1)-\b).$
\end{enumerate}
until the stopping criteria is met. \\
Output: $\r(k+1),\x(k+1),\y(k+1)$.
\caption{ADMM for $\ell_0$-LS after variable splitting}\label{alg:admm-l0}
\end{algorithm}

\subsection{Reformulation of  ADMM for $\ell_0$-LS after variable splitting}
The form of Alg. \ref{alg:admm-l0} can be used in practice directly, but there is no theoretical guarantee to the best of the authors' knowledge when hard thresholding operator is applied in step 2 and the roles of the three parameters $\alpha,\beta,\tau$ to the algorithm are not very clear, making them  inconvenient to be tuned. However, after some transformations, one can get a formula on $\x(k+1)$ which is relatively convenient to be analyzed and expresses clear roles of $\tau$, $\alpha\beta$ and $\beta$.

Firstly, a useful lemma is introduced as follows.
 \begin{lemma}
\label{lem:1}
For two series $\{a(k)\},\{b(k)\}$, where $a(k),b(k)\in\bbR, k\in\{0,1,2,\cdots\}$,  and  three numbers $c_1,c_2,c_3\in\bbR$, if $b(k+2)=c_1 b(k+1)+c_2 a(k+1)+c_3 a(k),k\in\{0,1,2,\cdots\}$, then for $k\in\{1,2,3,\cdots\}$, one has
\begin{equation}
b(k+1)=c_1^k b(1)+c_2a(k)+(c_1 c_2+c_3)\sum_{i=1}^{k-1}c_1^{k-1-i}a(i)+c_1^{k-1}c_3a(0).
\end{equation}
\end{lemma}
 The proof of Lemma \ref{lem:1} can be found in the supplementary.

 Consider the general CS model \eqref{eq:general_model}, denote
\be
\c(k)=\b-\A\x(k)=\A(\x_S-\x(k))+\e^{\prime}.\label{eq:11}
\ee
Then in step 1 of Alg. \ref{alg:admm-l0},  letting $\r(k+2)-\r(k+1)$ and using the identity in step 3 to eliminate the factors $\y(k+1),\y(k)$,
for $k\in\{0,1,2,\cdots\}$, we have
\be
\r(k+2)-\frac{1}{1+\alpha\beta}\r(k+1)=\frac{\alpha\beta}{1+\alpha\beta}\left(2\c(k+1)-\c(k)\right).\nonumber
\ee
In addition, in step 1 of Alg. \ref{alg:admm-l0}, $\r(1)=\frac{\alpha\beta}{1+\alpha\beta}\c(0)$. Then by Lemma \ref{lem:1} for $k\in\{1,2,3,\cdots\}$, one has
\be
\r(k+1)=\frac{2\alpha\beta}{1+\alpha\beta}\c(k)+\sum_{i=1}^{k-1}\frac{\alpha\beta(1-\alpha\beta)}{(1+\alpha\beta)^{k+1-i}}\c(i)-\frac{\alpha^2\beta^2}{(1+\alpha\beta)^{k+1}}c(0) \nonumber.%\label{eq:12}
\ee
Then for $k\in\{1,2,3,\cdots\}$, the step 2 of Alg. \ref{alg:admm-l0} can be reformulated as follows:
\begin{eqnarray}
\x(k+1)&=&\rm{hard}\Bigg(\x(k)+\frac{2\tau}{1+\alpha\beta}\A^{\! T}\c(k)+\tau(1-\alpha\beta)\sum_{i=1}^{k-1}\frac{1}{(1+\alpha\beta)^{k+1-i}}\A^{\! T}\c(i) \nonumber\\
&&-\frac{\alpha\beta\tau}{(1+\alpha\beta)^{k+1}}\A^{\! T}\c(0),\sqrt{\frac{2\tau}{\beta}}\Bigg), \label{eq:13}
\end{eqnarray}
In \eqref{eq:13}, one can see that $\x(k+1)$ has no dependence on the splitting variable $\r(k+1)$ or the dual variable $\y(k)$, so steps 1 and 3 can be eliminated in Alg. \ref{alg:admm-l0} if  \eqref{eq:13} is used to update $\x(k+1)$. From \eqref{eq:13}, the alternating direction iteration of ADMM can be seen as a method to take advantage of the approximation vector $\x(i),i\in\{0,1,2,\cdots,k-1\}$ before the $k$-iteration to update $\x(k+1)$ in some effective way (in each iteration, we minimize $\r(k+1)$ as well as $\x(k+1)$). In addition, the parameters $\alpha$ and $\beta$ influence the vector in the hard thresholding operator only by their product $\alpha\beta$, which is not obvious in the original formula in step 2 of Alg. \ref{alg:admm-l0}. $\alpha\beta$ determines the decay rate of the impact of the residue vector $\A^{\! T}\c(i), i\in\{0,1,2,\cdots,k\}$ as the iteration goes. For $\alpha\beta$ fixed, $\tau$ gives a tradeoff between the contributions of current approximation $\x(k)$ and the residue in the past iterations and $\beta$ is used to tune the number of the nonzero entries in $\x(k+1)$.

Now define $\gamma=\alpha\beta,\mu=\frac{1+\alpha\beta}{2}\tau$ and assume that $\gamma$ is fixed in each iteration and the sparsity $s$ is known in priori or estimated beforehand. If one always maintains $s$ nonzero entries in $\x(n+1)$  in each iteration (in this case  $\beta$ is determined implicitly),
then \eqref{eq:13} can be reformulated as follows,
\be
\x(k+1)=H_s\left(\x(k)+\mu\left(\A^{\! T}\c(k)+\frac{1-\gamma}{2}\sum_{i=1}^{k-1}\frac{1}{(1+\gamma)^{k-i}}\A^{\!
T}\c(i)-\frac{\gamma}{2(1+\gamma)^{k}}\A^{\! T}\c(0)\right)\right). \label{eq:14} \\
\ee
When using \eqref{eq:14} in our iteration, in order to avoid repetitive computations, for $k\in\{1,2,3,\cdots\}$, we set
\begin{eqnarray*}
\u(k)&=&\frac{1-\gamma}{2}\sum_{i=1}^{k-1}\frac{1}{(1+\gamma)^{k-i}}\A^{\!
T}\c(i),\\
\v(k)&=&\frac{\gamma}{2(1+\gamma)^{k}}\A^{\! T}\c(0).
\end{eqnarray*}
Define $f(\x)=\frac{1}{2}\|\A\x-\b\|_2^2$. Then one has
\begin{eqnarray*}
\u(k+1)&=&-\frac{1-\gamma}{2(1+\gamma)}\nabla f(\x(k))+\frac{1}{1+\gamma}\u(k),\\
\v(k+1)&=&\frac{1}{1+\gamma}\v(k),
\end{eqnarray*}
where $\nabla f(\x(k))=\A^{\! T}(\A\x(k)-\b)$.
For $k=0$, by the iteration steps of Alg. \ref{alg:admm-l0}, if maintaining $s$ nonzero entries in $\x(1)$, one has
\[
\x(1)=H_s(\x(0)+\frac{1}{2}\mu\A^{\! T}(\b-\A\x(0))).\]
According to the above derivations and  assumptions,  Alg. \ref{alg:admm-l0} can be reformulated as Alg. \ref{alg:iad} which is called ``iterative alternating direction'' (IAD), corresponding to IHT.
\begin{algorithm}[ht]
Input: $\y,\A,s,\mu,\gamma,\x(0)$.\\
Initialization: $\x(1)=H_s(\x(0)+\frac{1}{2}\mu\A^{\! T}(\b-\A\x(0))), \u(1)=\0, \v(1)=\frac{\gamma}{2(1+\gamma)}\A^{\! T}(\b-\A\x(0)).$ \\
Iteration: For $k=1,2,3,\cdots$, go through the following steps.
\begin{enumerate}
\item $\nabla f(\x(k))=\A^{\! T}(\A\x(k)-\b);$
\item $\x(k+1)=H_s(\x(k)+\mu(-\nabla f(\x(k))+\u(k)-\v(k)));$
\item $\u(k+1)=-\frac{1-\gamma}{2(1+\gamma)}\nabla f(\x(k))+\frac{1}{1+\gamma}\u(k);$
\item $\v(k+1)=\frac{1}{1+\gamma}\v(k).$
\end{enumerate}
until the stopping criterion is met. \\
Output: $\mathbf{x}(k+1)$.
\caption{Iterative alternating direction (IAD)}\label{alg:iad}
\end{algorithm}

If $\u(k)$ and $\v(k)$ are set to $\0$, IAD degrades to IHT. In fact, if setting $\gamma=1$, then $\u(k)=\0,$ and thus the only difference between IAD and IHT is $\v(k)=\frac{1}{2^{k+1}}\A^{\! T}\c(0)$ which decays at exponential rate and has little impact on $\x(k+1)$ as iteration goes. Therefore, IAD  can be considered as an alternating direction version of IHT. However, the empirical performance of IAD can be much better than that of IHT when $\gamma$ is set to some small value, such as $\gamma=0.1$. In this case, $\u(k)\neq\0$ can improve the effect of the hard thresholding noteworthily.
In addition, just like the role of $\mu$ in IHT, the selection of $\mu$ makes a big difference about the empirical performance of IAD. Finally, the requested additional computation for the added steps 3 and 4 are clearly marginal.

%Finally it should be noted that if $s$ nonzero entries need to be maintained in each iteration by changing the value of $\beta$ dynamically, using Alg. \ref{alg:iad} rather than Alg. \ref{alg:admm-l0} is necessary. The reason is that $\r(k+1)$ depends not only the product $\alpha\beta$, but also $\beta$ itself. If $\beta$ is changed, the vector in hard thresholding operator $H_s$ will also change, making the right selection of $\beta$ seemingly impossible.

Corresponding to the well-known NIHT and HTP, the variants ``normalized iterative alternating direction'' (NIAD) and ``alternating direction pursuit'' (ADP ) are given in Alg. \ref{alg:niad} and Alg. \ref{alg:adp} respectively. In  step 2 of NIAD, $\mu(k+1)$ is set to a step size that maximally reduces the error \cite{golub2012matrix} in each iteration. If $\x(0)$ is set to some $s$-sparse vector, $\mu(1)$ can be selected according to the initialization step; if $\x(0)$ is set to $\0$ simply, denote $S(1)$ the indices of the $s$ largest magnitude entries in $\A^{\! T}\b$,  one can set $\mu(1)=\frac{\|(\A^{\! T}\b)_{S(1)}\|_2^2}{\|\A(\A^{\! T}\b)_{S(1)}\|_2^2}$. See more discussions in \cite{blumensath2010normalized}. In the initialization step and step 4 of ADP, ADP sets $\mu=1$ in each iteration and solves a least squares problem on the support of $\w(k+1)$ for debiasing just like HTP does \cite{foucart2011hard}.

\begin{table}
\begin{tabular}{p{6.5cm}p{6.5cm}}
\begin{minipage}{6.5cm}
\begin{algorithm}[H]
Input: $\y,\A,s,\gamma,\x(0)$.\\
Initialization:\\ $\mu(1)=\frac{\|(\A^{\! T}(\b-\A\x(0)))_{\rm supp(\x(0))}\|_2^2}{\|\A(\A^{\! *}(\b-\A\x(0)))_{\rm supp(\x(0))}\|_2^2};$\\
$\x(1)=H_s(\x(0)+\frac{1}{2}\mu(1)\A^{\! T}(\b-\A\x(0)));$\\
 $\qquad\quad\u(1)=\0, \v(1)=\frac{\gamma}{2(1+\gamma)}\A^{\! T}(\b-\A\x(0)).$ \\
Iteration: For $k=1,2,3,\cdots$, go through the following steps.
\begin{enumerate}
\item $\nabla f(\x(k))=\A^{\! T}(\A\x(k)-\b);$
\item $\mu(k+1)=$
\\$\frac{\|\left(-\nabla f(\x(k))+\u(k)-\v(k)\right)_{S(k)}\|_2^2}{\|\A\left(-\nabla f(\x(k))+\u(k)-\v(k)\right)_{S(k)}\|_2^2};$
\item $\x(k+1)=H_s\big(\x(k)+\mu(k+1)$\\$\cdot(-\nabla f(\x(k))+\u(k)-\v(k))\big);$
\item $S(k+1)=\supp\left(\x(k+1)\right);$
\item $\u(k+1)=-\frac{1-\gamma}{2(1+\gamma)}\nabla f(\x(k))+\frac{1}{1+\gamma}\u(k);$
\item $\v(k+1)=\frac{1}{1+\gamma}\v(k).$
\end{enumerate}
until the stopping criterion is met. \\
Output: $\mathbf{x}(k+1)$.
\caption{Normalized iterative alternating direction (NIAD)}\label{alg:niad}
\end{algorithm}
\end{minipage}
&
\begin{minipage}{6.5cm}

\begin{algorithm}[H]
Input: $\y,\A,s,\gamma,\x(0)$.\\
Initialization: \\
$\w(1)=H_s(\x(0)+\frac{1}{2}\A^{\! T}(\b-\A\x(0))); $\\ $S(1)=\supp(\w(1));$\\
 $\x(k+1)=$\\$\text{arg}\min_{\z \in\bbR^n}\{ \Vert\b-\A\z\Vert_2,\;
    \text{supp}(\z)\subseteq S(1)\};$\\
 $\u(1)=\0, \v(1)=\frac{\gamma}{2(1+\gamma)}\A^{\! T}\left(\b-\A\x(0)\right).$ \\
Iteration: For $k=1,2,3,\cdots$, go through the following steps.
\begin{enumerate}
\item
$\nabla f(\x(k))=\A^{\! T}(\A\x(k)-\b);$
\item $\w(k+1)=$\\$H_s\left(\x(k)-\nabla f(\x(k))+\u(k)-\v(k)\right);$
\item $S(k+1)=\supp(\w(k+1));$
\item $\x(k+1)=\text{arg}\min_{\z \in\bbR^n}\{ \Vert\b-\A\z\Vert_2,\;
    \text{supp}(\z)\subseteq S(k+1)\}$;
\item $\u(k+1)=-\frac{1-\gamma}{2(1+\gamma)}\nabla f(\x(k))+\frac{1}{1+\gamma}\u(k);$
\item $\v(k+1)=\frac{1}{1+\gamma}\v(k).$
\end{enumerate}
until the stopping criterion is met. \\
Output: $\mathbf{x}(k+1)$.
\caption{Alternating direction pursuit (ADP)}\label{alg:adp}
\end{algorithm}
\end{minipage}
\end{tabular}
\end{table}
\section{Theoretical analysis}
% RIP¶¨Òå
% Lemma1, Lemma 2ÃèÊö
% IAD£¬NIAD£¬ADP¶¨Àí
This section highlights our theoretical results for IAD, NIAD and ADP. The proofs can be found in the supplementary.

\begin{definition}[\cite{candes2005decoding}]
\label{def:rip}
A matrix $\mathbf{\A}\in\mathbb{R}^{m\times n}$ is said to satisfy the $(\delta_s,s)$-order RIP if $\left|\frac{\|\A\x\|_2^2}{\|\x\|_2^2}-1\right|\le\delta_s$ for all $\x$ with $\|\x\|_0\le s$ and $\x\neq\0$.
%\begin{equation}
%(1-\delta)\Vert\mathbf{x}\Vert_2^2\le\Vert\mathbf{\A}\mathbf{x}\Vert_2^2\le(1+\delta)\Vert\mathbf{x}\Vert_2^2, \label{eq:origin_def}
%\end{equation}
%where $0\le\delta\le1$. The infimum of $\delta$, denoted by $\delta_s$, is called the \emph{restricted isometry constant (RIC)} of $\mathbf{\A}$.
\end{definition}

\begin{theorem}
\label{thm:1}
Consider the general CS model \eqref{eq:general_model}.
For each algorithm \emph{alg} from \{IAD, NIAD, ADP\}, if $\A$ satisfies $\rho^{alg}>\frac{|1-\gamma|}{\gamma}$, alg is guaranteed after $k$ iterations to return an approximation $\x(k+1)$ satisfying
\begin{eqnarray}
\|\x_S-\x(k+1)\|_2&\le&(c_1^{alg}(\lambda_1^{alg})^{k-2}+c_2^{alg}(\lambda_2^{alg})^{k-2}+c_3^{alg} b^k)\|\x_S-\x(0)\|_2\nonumber\\
&&+(c_4^{alg}+c_5^{alg}(\lambda_1^{alg})^{k-2}+c_6^{alg}(\lambda_2^{alg})^{k-2}+c_7^{alg}b^k)\|\e^{\prime}\|_2.
\end{eqnarray}
where
\begin{eqnarray*}
&&\rho^{IAD}=\frac{2(1-\sqrt{3}(|1-\mu|+\mu\delta_{3s}))}{\sqrt{3}\mu(1+\delta_{3s})},
b_1^{IAD}=\sqrt{3}(|1-\mu|+\mu\delta_{3s}),  b_2^{IAD}=\frac{\sqrt{3}\mu|1-\gamma|(1+\delta_{3s})}{2};
\end{eqnarray*}
\begin{eqnarray*}
&&
\rho^{NIAD}=\frac{2(1-(2\sqrt{3}+1)\delta_{3s})}{\sqrt{3}(1-\delta_{3s})^2(1+\delta_{3s})},
b_1^{NIAD}=\frac{2\sqrt{3}\delta_{3s}}{1-\delta_{3s}}, \quad b_2^{NIAD}=\frac{\sqrt{3}|1-\gamma|(1+\delta_{3s})}{2(1-\delta_{3s})};
\end{eqnarray*}
\begin{eqnarray*}
&&
\rho^{ADP}=\frac{\sqrt{2}(\sqrt{1-\delta_{3s}^2}-\sqrt{2}\delta_{3s})}{1+\delta_{3s}},
b_1^{ADP}=\sqrt{\frac{2\delta_{3s}^2}{1-\delta_{3s}^2}}, \quad b_2^{ADP}=|1-\gamma|\sqrt{\frac{1+\delta_{3s}}{2(1-\delta_{3s})}};
\end{eqnarray*}
\be
b=\frac{1}{1+\gamma}<1;\nonumber
\ee
\begin{eqnarray*}
\lambda_1^{alg}=\frac{(b+b_1^{alg})+\sqrt{(b-b_1^{alg})^2+4 bb_2^{alg}}}{2}<1,  \lambda_2^{alg}=\frac{(b+b_1^{alg})-\sqrt{(b-b_1^{alg})^2+4bb_2^{alg}}}{2}<1.
\end{eqnarray*}
The coefficients $c_1^{alg},c_2^{alg},\cdots,c_7^{alg}$ are positive constants or  $\cO(k)$ numbers and don't influence whether the corresponding algorithms converge or not, which will be given in the supplementary.
\end{theorem}

The above theorem says that the convergence rate of alg from \{IAD,NIAD,ADP\} is determined jointly by $\delta_{3s},\mu\; \rm and\; \gamma$. If $\gamma=1$, IAD, NIAD and ADP will converge if $\delta_{3s}<\frac{\sqrt{3}}{3}-|1-\mu|, \delta_{3s}<\frac{2\sqrt{3}-1}{11}\approx0.224,\delta_{3s}<\frac{\sqrt{3}}{3}\approx 0.5773$ respectively, which are equivalent to the theoretical guarantees of IHT, NIHT and HTP in order. If $\gamma\neq 1$, the bounds of the proposed algorithms on $\delta_{3s}$ will be stricter than the above bounds, but are still positive constants. In the exact reconstruction case, i.e, $\x$ is $s$-sparse and there is no noise, we have the following corollary.

% ËµÃ÷gammaµÄÈ¡ÖµµÄÓ°Ïì
%
% ËµÃ÷ºÍIHT£¬HTP¼°NIHTÖ®¼äÀíÂÛ±£Ö¤µÄ¹ØÏµ
% Ö¸³öÀíÂÛ±£Ö¤ºÍ¾­ÑéÐÔÄÜÖ®¼äµÄÁªÏµ
% ËµÃ÷ c_1,..,c_7ÔÚsupplementary¸ø³ö
%

\begin{corollary}
When $\x$ is $s$-sparse with support set $S$ and there is no noise, denote $x_{\min}=\min\{|x_i|,i\in S\}$ and $\lambda^{alg}=\max\{\lambda_1^{alg},\lambda_2^{alg},b\}$, where alg is from \{IAD,NIAD,ADP\}. Then alg will find  the support $S$ of $\x$ exactly after
\[
\frac{\ln (x_{\min}/\|\x(0)-\x\|_2)-\ln(c_1^{alg}+c_2^{alg}+c_3^{alg}b^2)}{\ln \lambda^{alg}}+3\]
iterations.
\end{corollary}
\section{Experiments}

%The theoretical guarantees in terms of RICs provide us the intuition of the worst-case
% performance of reconstruction algorithms. But it does not say much about empirical performance since the theoretical guarantee by RIP is very weak (see the ``strong
% phase transition curve'' in \cite{blanchard2011phase}) and there is no efficient way to verify the RIC condition.
% ËµÃ÷ÊµÑé·½·¨
%% ÔÚÃ¿¸±×ÓÍ¼ÖÐÒÔell_1×÷Îª²Î¿¼±ê×¼£¬ÒÀ´Î»æÖÆIHT£¬IAD£¬IHT^{1/3}, IAD^{1/3}, NIHT, NIAD, HTP, ADP. µÄÖØ½¨ÂÊÇúÏß
% ËµÃ÷muµÄÈ¡Öµ
% ËµÃ÷gammaµÄÈ¡Öµ
% ¾ØÕó¼°ÏòÁ¿ÉèÖÃ
% Í£Ö¹Ìõ¼þ
% ËµÃ÷²ÎÊýÑ¡ÔñÉÏÀíÂÛÓëÊµ¼ÊµÄ²»¶Ô³ÆÐÔ

% ÊµÑé½á¹ûËµÃ÷
%% IAD,NIAD,ADPÏà¶ÔIHT£¬NIHT£¬HTP¾ßÓÐÖØÒªÓÅÊÆ£»ADPºÍNIADÔÚCARSÏÂÄÜ½Ó½üell_1 ÓÅ»¯
%% gammaµÄÈ¡Öµ
% gammaµÄÑ¡È¡ÎÊÌâ
% Ç¿µ÷gammaµÄ×÷ÓÃ£¬¸º·´À¡£¬Ò²¾ÍÊÇÌ°ÐÄÊÊ¶È¡£
% ÓÃ·ÂÕæÍ¼ÓèÒÔËµÃ÷£¬ÈçgammaÔ½Ð¡£¬µü´ú´ÎÊýÔ½¶à£»gammaÔ½´ó£¬µü´ú´ÎÊýÔ½ÉÙ
% gamma¾ö¶¨ÁËÌ°ÐÄËã·¨µÄ¼¤½ø³Ì¶È£¬¹ý´óµÄgammaÊ¹µÃ$\x(k-1),\cdots,\x(0)µÄÓ°Ïì½µµÍ£¬Ôö¼ÓÁËx(k)+muA^{T}c(k) µÄÖØÒªÐÔ£¬¹ýÐ¡µÄgamma£¬Ê¹µÃ$\x(k-1),\cdots,\x(0)µÄÓ°Ïì¹ý´ó¡£ ÄÜ·ñÑ¡ÔñºÏÊÊµÄgamma, ´ïµ½Ç¡µ½ºÃ´¦µÄ¼¤½ø³Ì¶È£¿

%

 In this section, we show the empirical
 performance of IAD, NIAD and ADP by comparing the
 exact reconstruction rate with the corresponding IHT, NIHT and HTP algorithms.
 By comparing the maximal sparsity level of the underlying sparse signals at which the perfect reconstruction is ensured (\cite{dai2009subspace} called this point critical sparsity), accuracy of the reconstruction can be compared empirically.
 In each trial, we construct an $m\times n(m=200,n=1000)$ measurement matrix $\A$ with entries drawn independently from Gaussian distribution $\cN(0,\frac{1}{m})$. In addition, an $s$-sparse vector $\x$ whose support is chosen at random. CARS signals and Gaussian signals are considered. Each nonzero element of Gaussian signals is drawn from standard Gaussian distribution and that of CARS signals is  from the set $\{1,-1\}$ uniformly at random. The sparsity level  ranges in $[1,60]$ in CARS signal case and $[1,100]$ in Gaussian signal case.
  For each reconstruction algorithm,  1000 independent trials are performed and  the exact reconstruction rate is plotted  in $y$-axis as the sparsity $s$ changes in $x$-axis.
 % Because the value of $\mu$ makes a big difference to the empirical performance of IAD and IHT,
 For IAD and IHT,
   two representative $\mu$ are set : $\mu=1$ and $\mu=1/3$ in each trial and the corresponding algorithms are called $\rm IAD^{1}$,  $\rm IHT^{1}$ and $\rm IAD^{1/3}$, $\rm IHT^{1/3}$ respectively. For IAD, NIAD and ADP,  $\gamma=0.1$ is  set simply to show the improved empirical performance to the corresponding IHT, NIHT and HTP. One can tune $\gamma$ to acquire a relatively better empirical performance, which is not our focus here. Eight subfigures are plotted and each subfigure tests a couple of  algorithms from ($\rm IAD^{1}$, $\rm IHT^{1}$), ($\rm IAD^{1/3}$, $\rm IHT^{1/3}$), (NIAD, NIHT), (ADP, HTP) in  CARS  or Gaussian signal cases (the top four are CARS signal cases and the bottom four are  Gaussian signal cases). Meanwhile, a well-known implementation of $\ell_1$ minimization, $\ell_1$-magic (\url{http://users.ece.gatech.edu/~justin/l1magic/}), is used as a base-line algorithm to compare the performance among different couples of algorithms. In order to balance
time complexity and computational accuracy, the stopping criterion ``$k\ge400\; \rm or \; \frac{\|\y-\A\x(k+1)\|_2}{\|\y\|_2}\le10^{-6}$'' is used for each algorithm. However, we use a different criterion ``the support $S(k+1)$ of $\x(k+1)$ is equal to the support $S$ of the original $\x$''  to judge whether the reconstruction is exact or not. Because if $S(k+1)=S$, one can acquire the exact $\x$ by a posteriori least squares fit on $S(k+1)$.

In Fig. \ref{img:1}, one can see that $\rm IAD^{1}$, $\rm IAD^{1/3}$, NIAD and ADP all outperform the coressponding algorithms $\rm IHT^{1}$, $\rm IHT^{1/3}$, NIHT and HTP greatly. In addition, the empirical performance of $\rm IAD^{1/3}$, NIAD and ADP are much better than that of $\ell_1$ minimization in the Gaussian signal case and can approach that of $\ell_1$ minimization in the CARS signal case, which is rare for iterative greedy algorithms in compressed sensing.

\begin{figure}[ht]
  \centering
  \includegraphics[scale=0.32]{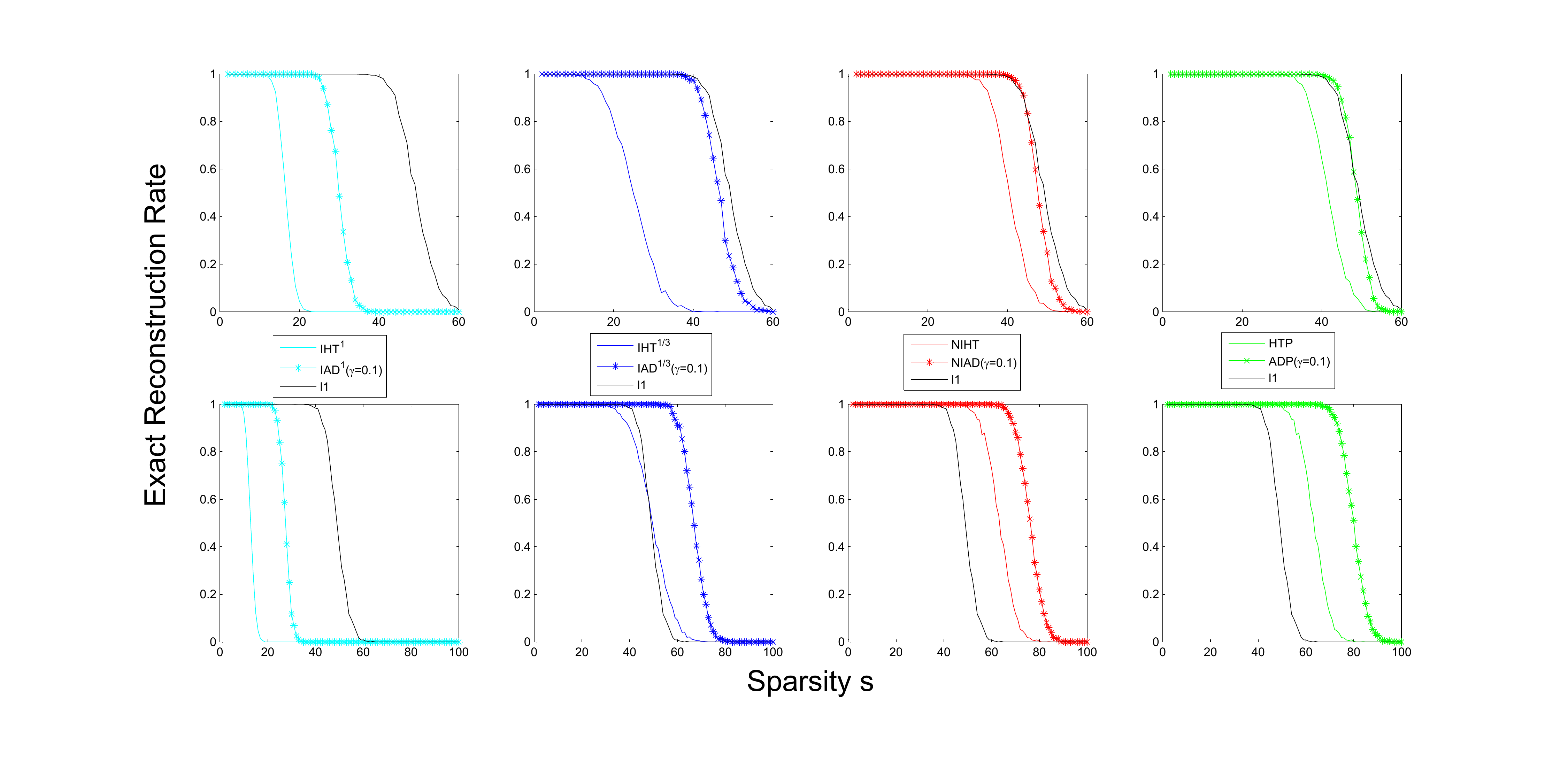}\\
  \caption{Comparison of $\rm IAD^{1}$ and $\rm IHT^{1}$, $\rm IAD^{1/3}$ and $\rm IHT^{1/3}$, NIAD and NIHT, ADP and HTP respectively from left to right (Top: CARS signals; bottom: Gaussian signals.).}\label{img:1}
\end{figure}

Table \ref{tb:1} lists the critical sparsity of the four pairs of algorithms  both in the Gaussian and CARS signal cases (row 2 and row 4 respectively). For each pair,  relative gains are calculated in both cases (row 3 and row 5 respectively). On the one hand, in Table \ref{tb:1}, one can know that relative to the corresponding ones, because the reconstruction capabilities of $\rm IHT^{1}$ and $\rm IHT^{1/3}$  are too restricted, $\rm IAD^{1}$ and $\rm IAD^{1/3}$ can acquire higher relative gains  than NIAD and ADP. While the improvements of NIAD and ADP relative to NIHT and HTP respectively are also substantial. On the other hand, in terms of critical sparsity, the variants NIAD and ADP have better empirical performance than $\rm IAD^{1}$ and $\rm IAD^{1/3}$, which is similar to the improved performance by NIHT and HTP relative to $\rm IHT^{1}$ and $\rm IHT^{1/3}$.

\begin{table*}[th]
\caption{The critical sparsity and relative gains of the four pairs of algorithms  in the CARS and Gaussian signal cases.}\label{tb:1}
\centering
\begin{tabular}{|c|c|c|c|c|c|c|c|c|}
\hline
Algorithms & $\rm IHT^{1}$ & $\rm IAD^{1}$ & $\rm IHT^{1/3}$ & $\rm IAD^{1/3}$ & NIHT & NIAD & HTP & ADP  \\

\hline
CARS Signals & 10 & 23 & 10 & 36 & 28 & 38 & 29 & 38  \\
\hline
Relative Gains & \multicolumn{2}{|c|}{130.0\%} &  \multicolumn{2}{|c|}{260.0\%} & \multicolumn{2}{|c|}{35.7\%} & \multicolumn{2}{|c|}{34.5\%}\\
\hline\hline
Gaussian Signals & 7 & 20 & 24 & 52 & 45 & 61 & 45 & 66 \\
\hline
Relative Gains & \multicolumn{2}{|c|}{185.7\%} &  \multicolumn{2}{|c|}{116.7\%} & \multicolumn{2}{|c|}{35.6\%} & \multicolumn{2}{|c|}{46.7\%} \\
\hline
\end{tabular}
\end{table*}

%
%
%\begin{table*}[t]
%\caption{The critical sparsity of algorithms under Gaussian measurement matrices with different sizes in the CARS signal case.}\label{tb:2}
%\centering
%\begin{tabular}{|c|c|c|c|c|c|c||c|c|c|c|c|c|c|}
%\hline
%Algorithms & OMP & $\ell_1$ & SP & CoSaMP & NIHT & HTP & STP1 & STP1.5 & STP2 & STP2.5 & STP3 & STP3.5 \\
%\hline
%Critical Sparsity(100$\times$3000) & 4 & \textbf{9} & 5 & 6 & 7 & 6 & 7 & 8 & 10 & \textbf{11} & 10 & 9 \\
%\hline
%Critical Sparsity(100$\times$1000) & 5 & \textbf{13} & 10 & 9 & 9 & 8 & 11 & 11 & 13 & \textbf{14} & 14 & 13 \\
%\hline
%Critical Sparsity(200$\times$1000) & 12 & \textbf{38} & 31 & 31 & 29 & 29 & 31 & 36 & 36 & \textbf{37}   & 37 & 35\\
%\hline
%Critical Sparsity(300$\times$1000) & 18 & \textbf{68} & 58 & 58 & 52 & 52 & 60 & 64 & \textbf{68} & 68   & 64 & 60\\
%\hline
%Critical Sparsity(400$\times$1000) & 23 & \textbf{110} & 95 & 95 & 71 & 80 & 98 & \textbf{104} & 104 & 104   & 104 & 98\\
%\hline
%Critical Sparsity(500$\times$1000) & 30 & \textbf{166} & 138 & 130 & 102 & 118 & 142 & \textbf{150} & 142 & 142   & 142 & 134\\
%\hline
%Critical Sparsity(600$\times$1000) & 42 & \textbf{234} & 190 & 182 & 142 & 150 & 194 & \textbf{198} & 194 & 190   & 190 & 186\\
%\hline
%\end{tabular}
%\end{table*}

% ·ÂÕæÉèÖÃ
% ×ÓÍ¼
% ±í¸ñ
%ÃèÊö
\section{Conclusions}
In this paper, we presented three alternating direction algorithms for $\ell_0$ regularization, called ``iterative alternating direction'' (IAD), ``normalized iterative alternating direction'' (NIAD) and ``alternating direction pursuit'' (ADP). They have provable theoretical guarantees and good empirical performance relative to the corresponding IHT, NIHT and HTP algorithms. However, the optimal value of $\gamma$ needs to be investigated further.

% \clearpage
\section*{Supplementary: the proofs of lemmas and theorems in main body}
\subsection{\label{proof:1}Proof of Lemma \ref{lem:1}}
%\begin{lemma}
%\label{lem:1}
%For two series $\{a(k)\},\{b(k)\}$, where $a(k),b(k)\in\bbR, k\in\{0,1,2,\cdots\}$,  and  three numbers $c_1,c_2,c_3\in\bbR$, if $b(k+2)=c_1 b(k+1)+c_2 a(k+1)+c_3 a(k),k\in\{0,1,2,\cdots\}$, then for $k\in\{1,2,3,\cdots\}$, one has
%\begin{equation}
%b(k+1)=c_1^k b(1)+c_2a(k)+(c_1 c_2+c_3)\sum_{i=1}^{k-1}c_1^{k-1-i}a(i)+c_1^{k-1}c_3a(0).
%\end{equation}
%\end{lemma}

The lemma is proved by recursively using the condition $b(k+2)=c_1 b(k+1)+c_2 a(k+1)+c_3 a(k),k\in\{0,1,2,\cdots\}$. For $k\in\{1,2,3,\cdots\}$, we have
\begin{eqnarray*}
b(k+1)&=&c_1 b(k)+c_2 a(k) +c_3 a(k-1) \\
&=&c_1(c_1 b(k-1)+c_2 a(k-1)+c_3 a(k-2))+c_2 a(k)+c_3 a(k-1)\\
&=&c_1^2b(k-1)+c_2a(k)+(c_1 c_2+c_3)a(k-1)+c_1 c_3a(k-2)\\
&=&c_1^2(c_1b(k-2)+c_2a(k-2)+c_3a(k-3))+c_2a(k)+(c_1 c_2+c_3)a(k-1)+c_1 c_3a(k-2)\\
&=&c_1^3b(k-2)+c_2a(k)+(c_1c_2+c_3)a(k-1)+c_1(c_1c_2+c_3)a(k-2)+c_1^2c_3a(k-3)\\
&\vdots\\
&=&c_1^k b(1)+c_2a(k)+(c_1 c_2+c_3)\sum_{i=1}^{k-1}c_1^{k-1-i}a(i)+c_1^{k-1}c_3a(0).
\end{eqnarray*}
\subsection{\label{lem:2}Lemma \ref{lem:2} and its proof}
The following lemma is critical for the proof of Theorem \ref{thm:1}.
\begin{lemma}
\label{lem:2}
For a series $\{a(k)\}$, $a(k)\ge0$,  for all $k\in\{0,1,2,\cdots\}$. If $a(k+1)\le b_1 a(k)+b_2 \sum_{i=1}^{k-1}b^{k-i}a(i)+b_3 b^k +b_4,k\in\{1,2,3,\cdots\}$, $b, b_1,b_2,b_3,b_4\ge0$, then for $k\in\{2,3,4,\cdots\}$,
\begin{eqnarray}
a(k+1)&\le &\left(\left(\omega_1+\frac{b_1}{2}\right)\lambda_1^{k-2}-\left(\omega_1-\frac{b_1}{2}\right)\lambda_2^{k-2}\right) a(2)\nonumber\\
 &&
+ \left(\left(\omega_2+\frac{b_2}{2}\right)\lambda_1^{k-2}-
\left(\omega_2-\frac{b_2}{2}\right)\lambda_2^{k-2}\right)
  a(1)\nonumber\\
 && +\left( b^{k-2}+ \left(\omega_1+\frac{b_1}{2}\right)\theta_1-\left(\omega_1-\frac{b_1}{2}\right)\theta_2\right)b_3 b^2\nonumber\\
 && +\left(1+\left(\omega_1+\frac{b_1}{2}\right)\frac{1-\lambda_1^{k-2}}{1-\lambda_1}-\left(\omega_1-\frac{b_1}{2}\right)\frac{1-\lambda_2^{k-2}}{1-\lambda_2}\right)b_4,\nonumber
\end{eqnarray}
where
\begin{eqnarray}
\lambda_1=\frac{(b+b_1)+\sqrt{(b-b_1)^2+4 bb_2}}{2}&,& \quad \lambda_2=\frac{(b+b_1)-\sqrt{(b-b_1)^2+4bb_2}}{2}; \nonumber\\
\theta_1=\left\{
  \begin{array}{cc}
    \frac{\lambda_1^{k-2}-b^{k-2}}{\lambda_1-b}, & \lambda_1\neq b \\
    (k-2)\lambda_1^{k-3}, & \lambda_1=b\\
  \end{array}
\right. &,&\quad
\theta_2=\left\{
  \begin{array}{cc}
    \frac{\lambda_2^{k-2}-b^{k-2}}{\lambda_2-b}, & \lambda_2\neq b \\
    (k-2)\lambda_2^{k-3}, & \lambda_2=b\\
  \end{array}
\right. ;\nonumber\\
\omega_1=\frac{-b b_1+b_1^2+2bb_2 }{2\sqrt{(b-b_1)^2+4 bb_2}}&,& \quad
\omega_2=\frac{(b+b_1)b_2}{2\sqrt{(b-b_1)^2+4 bb_2}}.\nonumber
\end{eqnarray}
$a(k+1)$ will be upper bounded if $(1-b)(1-b_1)>bb_2$ and $0<b<1$, and will approach $0$ as $k\rightarrow\infty$ if $(1-b)(1-b_1)>bb_2,0<b<1$ and $b_4=0$.
\end{lemma}

\begin{proof}
By recursively applying the condition $a(k+1)\le b_1 a(k)+b_2 \sum_{i=1}^{k-1}b^{k-i}a(i)+b_3 b^k +b_4,k\in\{2,3,4,\cdots\}$ to the highest-order term of  the right hand side, one has
\begin{eqnarray*}
a(k+1)&\le& b_1 a(k)+b_2 \sum_{i=1}^{k-1}b^{k-i}a(i)+b_3 b^k +b_4 \\
&\le& b_1(b_1 a(k-1)+b_2\sum_{i=1}^{k-2}b^{k-1-i}a(i)+b_3 b^{k-1}+b_4)+b_2\sum_{i=1}^{k-1}b^{k-i}a(i)+b_3 b^k +b_4\\
&=& (b_1^2+b_2 b)a(k-1)+(b_1 b_2+b_2 b)\sum_{i=1}^{k-2}b^{k-1-i}a(i)+(b_1 b_3+b_3 b)b^{k-1}+(b_1 +1)b_4\\
&\le& (b_1(b_1^2+b b_2)+b(b_1 b_2+b b_2))a(k-2)\\
&&+(b_2(b_1^2+b b_2)+b(b_1 b_2+b b_2))\sum_{i=1}^{k-3}b^{k-2-i}a(i)\\
&&+(b_3(b_1^2+b b_2)+b(b_1 b_3+b b_3))b^{k-2}+((b_1^2+b b_2)+b_1+1)b_4 \\
&\vdots&\\
&\le&c(j)a(j)+d(j)\sum_{i=1}^{j-1}b^{j-i}a(i)+e(j)b^j+l(j)b_4.
\end{eqnarray*}
When the highest-order is $j$ in the right hand, the above four series $\{c(j)\},\{d(j\},\{e(j)\},\{l(j)\}$ are the coefficients of $a(j),\sum_{i=1}^{j-1}b^{j-i}a(i), b^{j}, b_4$ respectively. According to the recursive procedure, one has
\be
c(k)=b_1,d(k)=b_2,e(k)=b_3,l(k)=b_4;\\
\ee
and
\begin{eqnarray}
c(j)&=&b_1 c(j+1)+bd(j+1), \label{eq:cj}\\
d(j)&=&b_2 c(j+1)+bd(j+1),  \label{eq:dj}\\
e(j)&=&b_3 c(j+1)+be(j+1),  \label{eq:ej}\\
l(j)&=&c(j+1)+l(j+1).   \label{eq:lj}
\end{eqnarray}
According to \eqref{eq:cj} and \eqref{eq:dj}, one has
\begin{equation}       %¿ªÊ¼ÊýÑ§»·¾³
\left[                %×óÀ¨ºÅ
  \begin{array}{c}   %¸Ã¾ØÕóÒ»¹²3ÁÐ£¬Ã¿Ò»ÁÐ¶¼¾ÓÖÐ·ÅÖÃ
    c(j) \\  %µÚÒ»ÐÐÔªËØ
    d(j) \\  %µÚ¶þÐÐÔªËØ
  \end{array}
\right] =   \left[                %×óÀ¨ºÅ
  \begin{array}{cc}   %¸Ã¾ØÕóÒ»¹²3ÁÐ£¬Ã¿Ò»ÁÐ¶¼¾ÓÖÐ·ÅÖÃ
    b_1 & b \\  %µÚÒ»ÐÐÔªËØ
    b_2 & b \\  %µÚ¶þÐÐÔªËØ
  \end{array}
\right]\left[                %×óÀ¨ºÅ
  \begin{array}{c}   %¸Ã¾ØÕóÒ»¹²3ÁÐ£¬Ã¿Ò»ÁÐ¶¼¾ÓÖÐ·ÅÖÃ
    c(j+1) \\  %µÚÒ»ÐÐÔªËØ
    d(j+1) \\  %µÚ¶þÐÐÔªËØ
  \end{array}
\right].       \label{eq:cjdj}       %ÓÒÀ¨ºÅ
\end{equation}
Define
\be
\B=\left[                %×óÀ¨ºÅ
  \begin{array}{cc}   %¸Ã¾ØÕóÒ»¹²3ÁÐ£¬Ã¿Ò»ÁÐ¶¼¾ÓÖÐ·ÅÖÃ
    b_1 & b \\  %µÚÒ»ÐÐÔªËØ
    b_2 & b\\  %µÚ¶þÐÐÔªËØ
  \end{array}
\right],\nonumber
\ee
By solving  the characteristic polynomial
\be
|\lambda\I-\B|=\left|                %×óÀ¨ºÅ
  \begin{array}{cc}   %¸Ã¾ØÕóÒ»¹²3ÁÐ£¬Ã¿Ò»ÁÐ¶¼¾ÓÖÐ·ÅÖÃ
    b_1 & b \\  %µÚÒ»ÐÐÔªËØ
    b_2 & b\\  %µÚ¶þÐÐÔªËØ
  \end{array}
\right|=\lambda^2-(b+b_1)\lambda+b b_1-bb_2=0, \nonumber
\ee
one has two eigenvalues
\be
\lambda_1=\frac{(b+b_1)+\sqrt{(b-b_1)^2+4b b_2}}{2}, \lambda_2=\frac{(b+b_1)-\sqrt{(b-b_1)^2+4b b_2}}{2},  \nonumber
\ee
and two  eigenvectors
\begin{eqnarray*}
\e_1=\left[                %×óÀ¨ºÅ
  \begin{array}{c}   %¸Ã¾ØÕóÒ»¹²3ÁÐ£¬Ã¿Ò»ÁÐ¶¼¾ÓÖÐ·ÅÖÃ
    b \\  %µÚÒ»ÐÐÔªËØ
    u+v \\  %µÚ¶þÐÐÔªËØ
  \end{array}
\right]&,&
\e_2=\left[                %×óÀ¨ºÅ
  \begin{array}{c}   %¸Ã¾ØÕóÒ»¹²3ÁÐ£¬Ã¿Ò»ÁÐ¶¼¾ÓÖÐ·ÅÖÃ
    b \\  %µÚÒ»ÐÐÔªËØ
    u-v \\  %µÚ¶þÐÐÔªËØ
  \end{array}
\right]\\
\end{eqnarray*}
corresponding to $\lambda_1$ and $\lambda_2$ respectively, where
\be
u=\frac{b-b_1}{2}\quad,\quad
v=\frac{\sqrt{(b-b_1)^2+4b b_2}}{2}.\nonumber
\ee
So $\B$ can be expressed by eigenvalue decomposition
\begin{eqnarray}
\B&=&[\e_1,\e_2]\left[                %×óÀ¨ºÅ
  \begin{array}{cc}   %¸Ã¾ØÕóÒ»¹²3ÁÐ£¬Ã¿Ò»ÁÐ¶¼¾ÓÖÐ·ÅÖÃ
    \lambda_1 & \\  %µÚÒ»ÐÐÔªËØ
    &\lambda_2\\  %µÚ¶þÐÐÔªËØ
  \end{array}
\right][\e_1,\e_2]^{-1} \nonumber\\
&=&\frac{1}{2bv}\left[                %×óÀ¨ºÅ
  \begin{array}{cc}   %¸Ã¾ØÕóÒ»¹²3ÁÐ£¬Ã¿Ò»ÁÐ¶¼¾ÓÖÐ·ÅÖÃ
    b & b\\  %µÚÒ»ÐÐÔªËØ
    u+v & u-v\\  %µÚ¶þÐÐÔªËØ
  \end{array}
\right]\left[                %×óÀ¨ºÅ
  \begin{array}{cc}   %¸Ã¾ØÕóÒ»¹²3ÁÐ£¬Ã¿Ò»ÁÐ¶¼¾ÓÖÐ·ÅÖÃ
    \lambda_1 & \\  %µÚÒ»ÐÐÔªËØ
    &\lambda_2\\  %µÚ¶þÐÐÔªËØ
  \end{array}
\right]\left[                %×óÀ¨ºÅ
  \begin{array}{cc}   %¸Ã¾ØÕóÒ»¹²3ÁÐ£¬Ã¿Ò»ÁÐ¶¼¾ÓÖÐ·ÅÖÃ
     -u+v & b\\
    u+v & -b\\  %µÚÒ»ÐÐÔªËØ
   %µÚ¶þÐÐÔªËØ
  \end{array}
\right].\label{eq:eigd}
\end{eqnarray}
According to \eqref{eq:cjdj} and \eqref{eq:eigd}, one has
\begin{eqnarray}
\left[                %×óÀ¨ºÅ
  \begin{array}{c}   %¸Ã¾ØÕóÒ»¹²3ÁÐ£¬Ã¿Ò»ÁÐ¶¼¾ÓÖÐ·ÅÖÃ
    c(i) \\  %µÚÒ»ÐÐÔªËØ
    d(i) \\  %µÚ¶þÐÐÔªËØ
  \end{array}
\right]&=&\B^{k-i}\left[                %×óÀ¨ºÅ
  \begin{array}{c}   %¸Ã¾ØÕóÒ»¹²3ÁÐ£¬Ã¿Ò»ÁÐ¶¼¾ÓÖÐ·ÅÖÃ
    c(k) \\  %µÚÒ»ÐÐÔªËØ
    d(k)\\  %µÚ¶þÐÐÔªËØ
  \end{array}
\right]\nonumber\\
&=&\frac{1}{2bv}\left[                %×óÀ¨ºÅ
  \begin{array}{cc}   %¸Ã¾ØÕóÒ»¹²3ÁÐ£¬Ã¿Ò»ÁÐ¶¼¾ÓÖÐ·ÅÖÃ
    b & b\\  %µÚÒ»ÐÐÔªËØ
    u+v & u-v\\  %µÚ¶þÐÐÔªËØ
  \end{array}
\right]\left[                %×óÀ¨ºÅ
  \begin{array}{cc}   %¸Ã¾ØÕóÒ»¹²3ÁÐ£¬Ã¿Ò»ÁÐ¶¼¾ÓÖÐ·ÅÖÃ
    \lambda_1^{k-i} & \\  %µÚÒ»ÐÐÔªËØ
    &\lambda_2^{k-i}\\  %µÚ¶þÐÐÔªËØ
  \end{array}
\right]\left[                %×óÀ¨ºÅ
  \begin{array}{cc}   %¸Ã¾ØÕóÒ»¹²3ÁÐ£¬Ã¿Ò»ÁÐ¶¼¾ÓÖÐ·ÅÖÃ
     -u+v & b\\
    u+v & -b\\  %µÚÒ»ÐÐÔªËØ
   %µÚ¶þÐÐÔªËØ
  \end{array}
\right]\left[
  \begin{array}{c}
    b_1 \\
    b_2\\
  \end{array}
\right]\nonumber\\
&=&\frac{1}{2bv}\left[                %×óÀ¨ºÅ
  \begin{array}{cc}   %¸Ã¾ØÕóÒ»¹²3ÁÐ£¬Ã¿Ò»ÁÐ¶¼¾ÓÖÐ·ÅÖÃ
    (-u+v)b\lambda_1^{k-i}+(u+v)b\lambda_2^{k-i} & b^2\lambda_1^{k-i}-b^2\lambda_2^{k-i}\\  % µÚÒ»ÐÐÔªËØ
    (-u^2+v^2)\lambda_1^{k-i}+(u^2-v^2)\lambda_2^{k-i}  &  (u+v)b\lambda_1^{k-i}-(u-v)b\lambda_2^{k-i}\\  % µÚ¶þÐÐÔªËØ
  \end{array}\right]\left[
  \begin{array}{c}
    b_1 \\
    b_2\\
  \end{array}
\right]\nonumber\\
&=&\frac{1}{2v}\left[                %×óÀ¨ºÅ
  \begin{array}{cc}   %¸Ã¾ØÕóÒ»¹²3ÁÐ£¬Ã¿Ò»ÁÐ¶¼¾ÓÖÐ·ÅÖÃ
     -(u-v)\lambda_1^{k-i}+(u+v)\lambda_2^{k-i} & b(\lambda_1^{k-i}-\lambda_2^{k-i})\\  % µÚÒ»ÐÐÔªËØ
    b_2(\lambda_1^{k-i}-\lambda_2^{k-i})  &  (u+v)\lambda_1^{k-i}-(u-v)\lambda_2^{k-i}\\  % µÚ¶þÐÐÔªËØ
  \end{array}\right]\left[
  \begin{array}{c}
    b_1 \\
    b_2\\
  \end{array}
\right]\nonumber\\
&=&\left[
    \begin{array}{c}
      \left(\omega_1+\frac{b_1}{2}\right)\lambda_1^{k-i}-\left(\omega_1-\frac{b_1}{2}\right)\lambda_2^{k-i}\\
      \left(\omega_2+\frac{b_2}{2}\right)\lambda_1^{k-i}-\left(\omega_2-\frac{b_2}{2}\right)\lambda_2^{k-i}
    \end{array}
\right],\label{eq:cidi}
\end{eqnarray}
where
\be
\omega_1=\frac{-b b_1+b_1^2+2bb_2 }{2\sqrt{(b-b_1)^2+4 b b_2}}, \quad
\omega_2=\frac{(b+b_1)b_2}{2\sqrt{(b-b_1)^2+4 b b_2}}. \nonumber
\ee
%\begin{eqnarray*}
%\theta_1(k)&=&\lambda_2^k-\lambda_1^k,\\
%\theta_2(k)&=&\lambda_2^k+\lambda_1^k.
%\end{eqnarray*}
%\begin{eqnarray*}
%\left[
%  \begin{array}{c}
%    c(i) \\
%    d(i)\\
%  \end{array}
%\right]&=&\frac{1}{2v}\left[
%  \begin{array}{cc}
%    -u\theta_1(k-i)+v\theta_2(k-i) & b\theta_1(k-i)  \\
%    b_2\theta_1(k-i) & u\theta_1(k-i)+v\theta_2(k-i)\\
%  \end{array}
%\right]\left[
%  \begin{array}{c}
%    b_1 \\
%    b_2\\
%  \end{array}
%\right]\\
%&=&\frac{1}{2v}\left[
%  \begin{array}{c}
%    (-u b_1+b_2 b)\theta_1(k-i)+vb_1\theta_2(k-i) \\
%    (b_1+u)b_2\theta_1(k-i)+v b_2\theta_2(k-i)\\
%  \end{array}
%\right]\\
%&=&\left[
%  \begin{array}{c}
%    \frac{-b b_1+b_1^2+2b_2 b}{2\sqrt{(b_1-b)^2+4b_2 b}}\theta_1(k-i)+\frac{b_1}{2}\theta_2 (k-i) \\
%    \frac{(b+b_1)b_2}{2\sqrt{(b_1-b)^2+4b_2 b}}\theta_1 (k-i)+\frac{b_2}{2}\theta_2(k-i)\\
%  \end{array}
%\right].
%\end{eqnarray*}
Recursively using \eqref{eq:ej}, one has
\begin{eqnarray}
e(2)&=&be(3)+b_3c(3)\nonumber\\
&=&b(be(4)+b_3c(4))+b_3c(3)\nonumber\\
&\vdots&\nonumber\\
&=&b^{k-2}e(k)+b_3\sum_{i=0}^{k-3}b^{i}c(i+3)\nonumber\\
&=&b_3 b^{k-2}+b_3\sum_{i=0}^{k-3}b^{i}\left(\left(\omega_1+\frac{b_1}{2}\right)\lambda_1^{k-i-3}-\left(\omega_1-\frac{b_1}{2}\right)\lambda_2^{k-i-3}\right)\nonumber\\
&=&b_3 b^{k-2}+b_3\left(\omega_1+\frac{b_1}{2}\right)\theta_1-b_3\left(\omega_1-\frac{b_1}{2}\right)\theta_2, \label{eq:e2}
\end{eqnarray}
where
\be
\theta_1=\left\{
  \begin{array}{cc}
    \frac{\lambda_1^{k-2}-b^{k-2}}{\lambda_1-b}, & \lambda_1\neq b \\
    (k-2)\lambda_1^{k-3}, & \lambda_1=b\\
  \end{array}
\right. ,\quad
\theta_2=\left\{
  \begin{array}{cc}
    \frac{\lambda_2^{k-2}-b^{k-2}}{\lambda_2-b}, & \lambda_2\neq b \\
    (k-2)\lambda_2^{k-3}, & \lambda_2=b\\
  \end{array}
\right..  \nonumber
\ee
Recursively using \eqref{eq:lj}, it follows that
\begin{eqnarray}
l(2)&=&c(3)+l(3) \nonumber\\
&=&c(3)+c(4)+\l(4) \nonumber\\
&\vdots& \nonumber\\
&=&l(k)+\sum_{i=0}^{k-3}c(i+3) \nonumber\\
&=&1+\sum_{i=0}^{k-3}\left(\left(\omega_1+\frac{b_1}{2}\right)\lambda_1^{k-i-3}-\left(\omega_1-\frac{b_1}{2}\right)\lambda_2^{k-i-3}\right) \nonumber\\
&=&1+\left(\omega_1+\frac{b_1}{2}\right)\frac{1-\lambda_1^{k-2}}{1-\lambda_1}-\left(\omega_1-\frac{b_1}{2}\right)\frac{1-\lambda_2^{k-2}}{1-\lambda_2}.
\label{eq:l2}
\end{eqnarray}
When the highest-order term is $j=2$, one has
\begin{eqnarray}
a(k+1)\le c(2)a(2)+d(2)a(1)+e(2)b^2+l(2)b_4, k\in\{2,3,4,\cdots\}. \label{eq:t30}
\end{eqnarray}

$a(k+1)$ will be upper bounded if $|\lambda_1|<1, |\lambda_2|<1, 0<b<1$. When $b,b_1,b_2\ge0$, it is equivalent to
\[(1-b_1)(1-b)>bb_2.\]
In addition, if $b_4=0$ holds additionally, $c(2),d(2),e(2),l(2)$ all will approach to 0 as $k\rightarrow\infty$. Applying \eqref{eq:cidi}, \eqref{eq:e2} and \eqref{eq:l2}to \eqref{eq:t30},
Lemma \ref{lem:2} is proved.
\end{proof}

% ¸ø³öÊÕÁ²ÐÔÌõ¼þËµÃ÷
%
\subsection{Some useful lemmas}
 The following three lemmas are used in the derivations of RIC related results.
\begin{lemma}[Consequences of the RIP]$ $
\label{lem:rip}
\begin{enumerate}
\item (Monotonicity \cite{candes2005decoding}) For any two positive integers $s\le s^{\prime}$, \quad $\delta_s\le\delta_{s^{\prime}}.$
\item For two vectors $\mathbf{p}, \mathbf{q}\in\mathbb{R}^{n}$ and $\mu>0$, if $|$\supp($\mathbf{p}$)$\cup$\supp($\mathbf{q}$)$|$$\le t$, then
\begin{eqnarray}
        |\langle\mathbf{p}, (\mathbf{I}-\mu\mathbf{\A}^{\! T}\mathbf{\A})\mathbf{q}\rangle|\le(|\mu-1|+\mu\delta_t)\Vert\mathbf{p}\Vert_2 \Vert\mathbf{q}\Vert_2;\label{rip11}
\end{eqnarray}
      moreover, if $U\subseteq\{1,\dots,n\}$ and $|U \cup \supp(\mathbf{q})$$|$$\le t$, then
\begin{eqnarray}
            \Vert((\mathbf{I}-\mu\mathbf{\A}^{\! T}\mathbf{\A})\mathbf{q})_U\Vert_2\le(|\mu-1|+\mu\delta_t)\Vert\mathbf{q}\Vert_2.\label{rip12}
\end{eqnarray}
We omit  the proofs of \eqref{rip11} and \eqref{rip12} here for their  similarity to the proofs of \cite[Lemma 1]{song2013improved}.
\end{enumerate}
\end{lemma}
\vspace{0.1in}
\begin{lemma}[Noise perturbation in partial support \cite{foucart2011hard}]
\label{lem:noise}
For the general CS model $
\mathbf{b}=\mathbf{\A}\mathbf{x}_S+\mathbf{e}^{\prime}$ in (\ref{eq:general_model}), letting $U\subseteq\{1,\ldots,n\}$ and $|U|\le u$, we have
\begin{eqnarray}
\label{rip13}
\Vert(\mathbf{\A}^{\! T}\mathbf{e}^{\prime})_{U}\Vert_2\le\sqrt{1+\delta_{u}}\Vert\mathbf{e}^{\prime}\Vert_2.
\end{eqnarray}
\end{lemma}
The next lemma introduces a simple inequality introduced in \cite{song2013improved} which is useful in our derivations.

Consider the general CS model $
\mathbf{b}=\mathbf{\A}\mathbf{x}_S+\mathbf{e}^{\prime}$ in (\ref{eq:general_model}). Let $S^{\prime}\subseteq \{1,2,\ldots,n\}$ and $|S^{\prime}|=t$. Let $\mathbf{z}_{p}$ be the solution of the least squares problem $\mbox{arg}\min_{\mathbf{z}\in\mathbb{R}^{n}}\{\Vert\mathbf{b}-
\mathbf{\A}\mathbf{z}\Vert_2, \;\text{supp}(\mathbf{z}$)$\subseteq S^{\prime}\}$. The least squares problem has the following orthogonal properties introduced in \cite{song2013improved}.

\begin{lemma}[Consequences for orthogonality by the RIP \cite{song2013improved}]
\label{lem:orthogonality-rip}
If $\delta_{s+t}<1$,
\begin{equation}
\Vert\mathbf{x}_S-\mathbf{z}_{p}\Vert_2\le\sqrt{\dfrac{1}{1-\delta_{s+t}^2}}
\Vert(\mathbf{x}_S)_{\overline{S^{\prime}}}\Vert_2+\dfrac{\sqrt{1+\delta_{t}}}
{1-\delta_{s+t}}\Vert\mathbf{e}^{\prime}\Vert_2.\label{eq:orthogonality-rip2}
\end{equation}
\end{lemma}

\subsection{Proof of Theorem 1}
% In this section, we prove
% ÏÈ¸ø³öÒ»¸öÒýÀí
% È»ºóËµÃ÷muÔÚ²»Í¬Ëã·¨ÏÂµÄÈ¡Öµ
% µ¥¶À¶ÔIAD,NIAD½éÉÜ
% ÔÚ¶ÔIHT½éÉÜ£¬×îºó½áºÏ½éÉÜ¸ø³ö¸÷¸ö²ÎÊý£¬
% È»ºóÒýÓÃÒýÀí2£¬¸ø³ö¶¨Àí1
%

Firstly, an inequality is derived from the equality \eqref{eq:14}, which uses a similar derivation from \cite{foucart2011hard}. Then we apply it to IAD, NIAD and ADP respectively. Finally, Lemma 2 is used to get the Theorem 1.

For $k\in\{1,2,3,\cdots\}$, using \eqref{eq:11}, \eqref{eq:14} can be rewrote as follows
\begin{eqnarray}
\x(k+1)&=&H_s\Bigg(\x(k)+\mu\Bigg(\A^{\! T}\A(\x_S-\x(k))+\frac{1-\gamma}{2}\sum_{i=1}^{k-1}\frac{1}{(1+\gamma)^{k-i}}\A^{\! T}\A(\x_S-\x(i))\nonumber\\
&&
-\frac{\gamma}{2(1+\gamma)^{k}}\A^{\! T}\A(\x_S-\x(0))+\left(1-\frac{1}{(1+\gamma)^{k+1}}\right)\A^{\! T}\e^{\prime}\Bigg)\Bigg).\label{eq:t1}
\end{eqnarray}

%\begin{eqnarray}
%\r(k+1)&=&\frac{2\gamma}{1+\gamma}\A(\x_S-\x(k))+\sum_{i=1}^{k-1}\frac{\gamma(1-\gamma)}{(1+\gamma)^{k+1-i}}\A(\x_S-x(i))\nonumber\\
%&&-\frac{\gamma^2}{(1+\gamma)^{k+1}}\A\x
%+(1-\frac{1}{(1+\gamma)^{k+1}})\e^{\prime}, (k\ge 1),
%\end{eqnarray}
Denote
\begin{eqnarray*}
\h(k)=\frac{1-\gamma}{2}\sum_{i=1}^{k-1}\frac{1}{(1+\gamma)^{k-i}}\A^{\! T}\A(\x_S-\x(i))-\frac{\gamma}{2(1+\gamma)^{k}}\A^{\! T}\A(\x_S-\x(0))+\left(1-\frac{1}{(1+\gamma)^{k+1}}\right)\A^{\! T}\e^{\prime},
\end{eqnarray*}
following a similar derivation skill from \cite{foucart2011hard}, in the hard thresholding operator \eqref{eq:t1}, one has
\begin{eqnarray*}
&&\|\left(\x(k)+\mu\left(\A^{\! T}\A(\x_S-\x(k))+\h(k)\right)\right)_S\|_2\nonumber\\
&&\quad\le\|\left(\x(k)+\mu\left(\A^{\! T}\A(\x_S-\x(k))+\h(k)\right)\right)_{S(k+1)}\|_2, \\
\end{eqnarray*}
then
\begin{eqnarray}
&&\|\left(\x(k)+\mu\left(\A^{\! T}\A(\x_S-\x(k))+\h(k)\right)\right)_{S\backslash S(k+1) }\|_2\nonumber\\
&&\quad\le\|\left(\x(k)+\mu\left(\A^{\! T}\A(\x_S-\x(k))+\h(k)\right)\right)_{S(k+1)\backslash S}\|_2. \label{eq:proof1}
\end{eqnarray}

For the right hand of \eqref{eq:proof1}, one has
\begin{eqnarray}
&&\|\left(\x(k)+\mu\left(\A^{\! T}\A(\x_S-\x(k))+\h(k)\right)\right)_{S(k+1)\backslash S}\|_2\nonumber\\
&&\quad = \|((\mu\A^{T}\A-\I)(\x_S-\x(k))+\mu\h(k))_{S(k+1)\backslash S}\|_2. \label{eq:proof2}
\end{eqnarray}
For the left hand of \eqref{eq:proof1}, one has
\begin{eqnarray}
&&\|\left(\x(k)+\mu\left(\A^{\! T}\A(\x_S-\x(k))+\h(k)\right)\right)_{S\backslash S(k+1) }\|_2 \nonumber\\
&=&\|\left(\x(k)+\mu\left(\left(\A^{\! T}\A(\x_S-\x(k))+\h(k)\right)\right)-\x_S+\x_S
\right)_{S\backslash S(k+1) }\|_2 \nonumber\\
&\ge& \|(\x_S)_{\overline{S(k+1)}}\|_2-\|((\mu\A^{T}\A-\I)(\x_S-\x(k))+\mu\h(k))_{S\backslash S(k+1)}\|_2.\label{eq:proof3}
\end{eqnarray}
Denote $S\triangle S(k+1)=(S\backslash S(k+1))\cup (S(k+1)\backslash S)$. Combing \eqref{eq:proof2} and \eqref{eq:proof3}, it follows that
\begin{eqnarray}
\|(\x_S)_{\overline{S(k+1)}}\|_2\le\sqrt{2}\|((\mu\A^{\! T}\A-\I)(\x_S-\x(k))+\mu\h(k))_{S\triangle S(k+1)}\|_2. \label{eq:t2}
\end{eqnarray}
\begin{enumerate}
\item
For IAD and NIAD, following a similar derivation skill from \cite{foucart2011hard}, one has
\begin{eqnarray*}
\|\x_S-\x(k+1)\|_2^2&=&\|(\x_S-\x(k+1))_{S(k+1)}\|_2^2+\|(\x_S)_{\overline{S(k+1)}}\|_2^2 \\
&=&\|((\mu\A^{\! T}\A-\I)(\x_S-\x(k))+\mu\h(k))_{S(k+1)}\|_2^2\\
&&+2\|((\mu\A^{\! T}\A-\I)(\x_S-\x(k))+\mu\h(k))_{S\triangle S(k+1)}\|_2^2 \\
&\le&3\|((\mu\A^{\! T}\A-\I)(\x_S-\x(k))+\mu\h(k))_{S\cup S(k+1)}\|_2^2.
\end{eqnarray*}
By triangle inequality, RIC definition and Lemma \ref{lem:rip}, it follows that
\begin{eqnarray}
&&\|\x_S-\x(k+1)\|_2 \nonumber\\
&\le&\sqrt{3}\Bigg((|\mu-1|+\mu\delta_{3s})\|\x_S-\x(k)\|_2\nonumber\\
&&+\frac{\mu|1-\gamma|}{2}(1+\delta_{3s})\sum_{i=1}^{k-1}\frac{1}{(1+\gamma)^{k-i}}\|\x_S-\x(i)\|_2\nonumber\\
&&
+\frac{\mu\gamma(1+\delta_{3s})}{(1+\gamma)^k}\|\x_S-\x(0)\|_2+\mu\left(1-\frac{1}{(1+\gamma)^{k+1}}\right)\sqrt{1+\delta_{3s}}\|\e^{\prime}\|_2\Bigg).
\label{eq:t4}
\end{eqnarray}
Particularly, for $k=1$, one has
\begin{eqnarray}
&&\|\x_S-\x(2)\|_2 \nonumber\\
&\le&\sqrt{3}\Bigg((|\mu-1|+\mu\delta_{3s})\|\x_S-\x(1)\|_2\nonumber\\
&&
+\frac{\mu\gamma(1+\delta_{3s})}{1+\gamma}\|\x_S-\x(0)\|_2+\mu\left(1-\frac{1}{(1+\gamma)^{2}}\right)\sqrt{1+\delta_{3s}}\|\e^{\prime}\|_2\Bigg).
\label{eq:t14}
\end{eqnarray}

For NIAD, in step 2, by RIP definition,
\[
\frac{1}{1+\delta_{3s}}\le\mu(k+1)=\frac{\|\left(-\nabla f(\x(k))+\u(k)-\v(k)\right)_{S(k)}\|_2^2}{\|\A\left(-\nabla f(\x(k))+\u(k)-\v(k)\right)_{S(k)}\|_2^2}\le\frac{1}{1-\delta_{3s}}.\]
Then one has
\begin{eqnarray}
|\mu-1|+\mu\delta_{3s}&\le&\max\left\{\left|\frac{1}{1+\delta_{3s}}-1\right|+\frac{1}{1+\delta_{3s}}\delta_{3s},
\left|\frac{1}{1-\delta_{3s}}-1\right|+\frac{1}{1-\delta_{3s}}\delta_{3s}\right\}\nonumber\\
&=&\max\left\{\frac{2\delta_{3s}}{1+\delta_{3s}},\frac{2\delta_{3s}}{1-\delta_{3s}}\right\}\nonumber\\
&=&\frac{2\delta_{3s}}{1-\delta_{3s}}. \label{eq:t5}
\end{eqnarray}
So, for NIAD, from \eqref{eq:t4} and \eqref{eq:t5}, it follows that
\begin{eqnarray}
&&\|\x_S-\x(k+1)\|_2 \nonumber\\
&\le&\sqrt{3}\Bigg(\frac{2\delta_{3s}}{1-\delta_{3s}}\|\x_S-\x(k)\|_2+\frac{|1-\gamma|(1+\delta_{3s})}{2(1-\delta_{3s})}\sum_{i=1}^{k-1}\frac{1}{(1+\gamma)^{k-i}}\|\x_S-\x(i)\|_2\nonumber\\
&&
+\frac{\gamma(1+\delta_{3s})}{(1+\gamma)^k(1-\delta_{3s})}\|\x_S-\x(0)\|_2\nonumber\\
&&+\left(1-\frac{1}{(1+\gamma)^{k+1}}\right)\frac{\sqrt{1+\delta_{3s}}}{1-\delta_{3s}}\|\e^{\prime}\|_2\Bigg).
\label{eq:t6}
\end{eqnarray}

Particularly, for $k=1$, one has
\begin{eqnarray}
&&\|\x_S-\x(2)\|_2 \nonumber\\
&\le&\sqrt{3}\Bigg(\frac{2\delta_{3s}}{1-\delta_{3s}}\|\x_S-\x(1)\|_2
+\frac{\gamma(1+\delta_{3s})}{(1+\gamma)(1-\delta_{3s})}\|\x_S-\x(0)\|_2\nonumber\\
&&+\left(1-\frac{1}{(1+\gamma)^{2}}\right)\frac{\sqrt{1+\delta_{3s}}}{1-\delta_{3s}}\|\e^{\prime}\|_2\Bigg).
\label{eq:t16}
\end{eqnarray}

%
%\begin{eqnarray}
%&&\|\x_S-\x(k+1)\|_2\le\sqrt{3}\Bigg((|\mu-1|+\mu\delta_{3s})\|\x_S-\x(k)\|_2 \nonumber\\
%&&+\frac{\mu|1-\gamma|}{2}(1+\delta_{3s})\sum_{i=1}^{k-1}\frac{1}{(1+\gamma)^{k-i}}\|\x_S-\x(i)\|_2\nonumber\\
%&&+\frac{\mu\gamma(1+\delta_{3s})}{(1+\gamma)^k}\|\x_S-\x(0)\|_2+\mu\left(1-\frac{1}{(1+\gamma)^{k+1}}\right)\sqrt{1+\delta_{3s}}\|\e^{\prime}\|_2\Bigg).
%\end{eqnarray}

\item For ADP, in step 4, by Lemma \ref{lem:orthogonality-rip},
\be
\|\x_S-\x(k+1)\|_2\le\sqrt{\frac{1}{1-\delta_{3s}^2}}\|(\x_S)_{\overline{S(k+1)}}\|_2+\frac{\sqrt{1+\delta_{3s}}}{1-\delta_{3s}}\|\e^{\prime}\|_2.
\label{eq:t7}
\ee
Setting $\mu=1$ and combing \eqref{eq:t2} and \eqref{eq:t7}, one has
\begin{eqnarray}
&&\|\x_S-\x(k+1)\|_2\nonumber\\
&\le&\sqrt{\frac{2}{1-\delta_{3s}^2}}\|((\A^{\! T}\A-\I)(\x_S-\x(k))+\h(k))_{S\triangle S(k+1)}\|_2+\frac{\sqrt{1+\delta_{3s}}}{1-\delta_{3s}}\|\e^{\prime}\|_2\nonumber\\
&\le&\sqrt{\frac{2}{1-\delta_{3s}^2}}\Bigg(\delta_{3s}\|\x_S-\x(k)\|_2+\frac{|1-\gamma|}{2}(1+\delta_{3s})\sum_{i=1}^{k-1}\frac{1}{(1+\gamma)^{k-i}}\|\x_S-\x(i)\|_2\nonumber\\
&&+\frac{\gamma(1+\delta_{3s})}{(1+\gamma)^k}\|\x_S-\x(0)\|_2+\left(\sqrt{\frac{1+\delta_{3s}}{2(1-\delta_{3s})}}+1-\frac{1}{(1+\gamma)^{k+1}}\right)\sqrt{1+\delta_{3s}}\|\e^{\prime}\|_2\Bigg).
\nonumber\\ \label{eq:t16}
\end{eqnarray}

Particularly, for $k=1$, one has
\begin{eqnarray}
&&\|\x_S-\x(2)\|_2\nonumber\\
&\le&\sqrt{\frac{2}{1-\delta_{3s}^2}}\Bigg(\delta_{3s}\|\x_S-\x(1)\|_2+\frac{\gamma(1+\delta_{3s})}{1+\gamma}\|\x_S-\x(0)\|_2\nonumber\\
&&+\left(\sqrt{\frac{1+\delta_{3s}}{2(1-\delta_{3s})}}+1-\frac{1}{(1+\gamma)^{2}}\right)\sqrt{1+\delta_{3s}}\|\e^{\prime}\|_2\Bigg).
\label{eq:t17}
\end{eqnarray}

In \eqref{eq:t4}, denote
\begin{eqnarray*}
b_1^{IAD}&=&\sqrt{3}(|\mu-1|+\mu\delta_{3s}),\\
b_2^{IAD}&=&\frac{\sqrt{3}\mu|1-\gamma|(1+\delta_{3s})}{2},\\
b_3^{IAD}&=&b_5\|\x_S-\x(0)\|_2-b_6\|\e^{\prime}\|_2,\\
b_4^{IAD}&=&b_7^{IAD}\|\e^{\prime}\|_2.
\end{eqnarray*}
where
\begin{eqnarray*}
b_5^{IAD}&=&\sqrt{3}\mu\gamma(1+\delta_{3s}),\\
b_6^{IAD}&=&\frac{\mu\sqrt{3(1+\delta_{3s})}}{1+\gamma},\\
b_7^{IAD}&=&\mu\sqrt{3(1+\delta_{3s})}.
\end{eqnarray*}

In \eqref{eq:t6}, denote
\begin{eqnarray*}
b_1^{NIAD}&=&\frac{2\sqrt{3}\delta_{3s}}{1-\delta_{3s}},\\
b_2^{NIAD}&=&\frac{\sqrt{3}|1-\gamma|(1+\delta_{3s})}{2(1-\delta_{3s})},\\
b_3^{NIAD}&=&b_5^{NIAD}\|\x_S-\x(0)\|_2-b_6^{NIAD}\|\e^{\prime}\|_2,\\
b_4^{NIAD}&=&b_7^{NIAD}\|\e^{\prime}\|_2.
\end{eqnarray*}
where
\begin{eqnarray*}
b_5^{NIAD}&=&\frac{\sqrt{3}\gamma(1+\delta_{3s})}{1-\delta_{3s}},\\
b_6^{NIAD}&=&\frac{\sqrt{3(1+\delta_{3s})}}{(1+\gamma)(1-\delta_{3s})},\\
b_7^{NIAD}&=&\frac{\sqrt{3(1+\delta_{3s})}}{1-\delta_{3s}}.
\end{eqnarray*}

In \eqref{eq:t16}, denote
\begin{eqnarray*}
b_1^{ADP}&=&\sqrt{\frac{2\delta_{3s}^2}{1-\delta_{3s}^2}}\quad,\\
b_2^{ADP}&=&|1-\gamma|\sqrt{\frac{1+\delta_{3s}}{2(1-\delta_{3s})}}\quad,\\
b_3^{ADP}&=&b_5^{ADP}\|\x_S-\x(0)\|_2-b_6^{ADP}\|\e^{\prime}\|_2,\\
b_4^{ADP}&=&b_7^{ADP}\|\e^{\prime}\|_2.
\end{eqnarray*}
where
\begin{eqnarray*}
b_5^{ADP}&=&\gamma\sqrt{\frac{2(1+\delta_{3s})}{1-\delta_{3s}}},\\
b_6^{ADP}&=&\frac{1}{1+\gamma}\sqrt{\frac{2}{1-\delta_{3s}}},\\
b_7^{ADP}&=&\frac{\sqrt{1+\delta_{3s}}}{1-\delta_{3s}}+\sqrt{\frac{2}{1-\delta_{3s}}}.
\end{eqnarray*}

In the initialization step, i.e., when $k=0$, the derivation step is similar to that in \cite{foucart2011hard}. We omit the steps here.

For IAD,
\begin{eqnarray}
\|\x_S-\x(1)\|_2&\le&\sqrt{3}\left(\left|\frac{1}{2}\mu-1\right|+\frac{1}{2}\mu\delta_{3s}\right)\|\x_S\|_2+\frac{1}{2}\mu\sqrt{3(1+\delta_{3s})}\|\e^{\prime}\|_2 . \label{eq:t8}
\end{eqnarray}

For NIAD, from \eqref{eq:t6}, in order to guarantee the convergence of NIAD, one can see that a necessary condition is $b_1^{NIAD}<1$, which results in  $\delta_{3s}<0.224$. Thus, by \eqref{eq:t5},
\begin{eqnarray*}
\left|\frac{1}{2}\mu-1\right|+\frac{1}{2}\mu\delta_{3s}&\le&\max\left\{\left|\frac{1}{2(1+\delta_{3s})}-1\right|+\frac{\delta_{3s}}{2(1+\delta_{3s})},\left|\frac{1}{2(1-\delta_{3s})}-1\right|+\frac{\delta_{3s}}{2(1-\delta_{3s})}\right\}\\
&=&\max\left\{\frac{1+3\delta_{3s}}{2(1+\delta_{3s})},\left|\frac{-1+2\delta_{3s}}{2(1-\delta_{3s})}\right|+\frac{\delta_{3s}}{2(1-\delta_{3s})}\right\}\\
&=&\max\left\{\frac{1}{2}+\frac{\delta_{3s}}{1+\delta_{3s}},\frac{1}{2}\right\}\\
&=&\frac{1+3\delta_{3s}}{2(1+\delta_{3s})}. \label{eq:t20}
\end{eqnarray*}

Then from \eqref{eq:t8} and \eqref{eq:t20}, one has
\begin{eqnarray}
\|\x_S-\x(1)\|_2&\le&\frac{\sqrt{3}(1+3\delta_{3s})}{2(1+\delta_{3s})}\|\x_S-\x(0)\|_2+\frac{\sqrt{3(1+\delta_{3s})}}{2(1-\delta_{3s})}\|\e^{\prime}\|_2.
\label{eq:t10}
\end{eqnarray}

For ADP,
\begin{eqnarray*}
\|\x_S-\x(1)\|_2\le\sqrt{\frac{2\delta_{3s}^2}{1-\delta_{3s}^2}}\|\x_S-\x(0)\|_2+\left(\frac{1}{\sqrt{2(1-\delta_{3s})}}+\frac{\sqrt{1+\delta_{3s}}}{1-\delta_{3s}}\right)\|\e^{\prime}\|_2.
\label{eq:t11}
\end{eqnarray*}

In \eqref{eq:t8}, denote
\be
b_8^{IAD}=\sqrt{3}\left(\left|\frac{1}{2}\mu-1\right|+\frac{1}{2}\mu\delta_{3s}\right),\quad b_9^{IAD}=\frac{1}{2}\mu\sqrt{3(1+\delta_{3s})}. \nonumber
\ee

In \eqref{eq:t10}, denote
\be
b_8^{NIAD}=\frac{\sqrt{3}(1+3\delta_{3s})}{2(1+\delta_{3s})}, \quad
b_9^{NIAD}=\frac{\sqrt{3(1+\delta_{3s})}}{2(1-\delta_{3s})}. \nonumber
\ee

In \eqref{eq:t11}, denote
\be
b_8^{ADP}=\sqrt{\frac{2\delta_{3s}^2}{1-\delta_{3s}^2}},\quad b_9^{ADP}=\frac{1}{\sqrt{2(1-\delta_{3s})}}+\frac{\sqrt{1+\delta_{3s}}}{1-\delta_{3s}}. \nonumber
\ee

Denote $a(k)=\|\x_S-\x(k)\|_2, b=\frac{1}{1+\gamma}$ and alg as any one from \{IAD, NIAD, ADP\}. From above derivations and notations, it follows that,
\begin{eqnarray}
a(k+1)&\le& b_1^{alg}a(k)+b_2^{alg}\sum_{i=1}^{k-1}b^{k-i}a(i)+b_3^{alg}b^k+b_4^{alg}, \label{eq:t21}\\
a(2)&\le& b_1^{alg}a(1)+b_3^{alg}b+b_4^{alg}, \label{eq:t22}\\
a(1)&\le& b_8^{alg}a(0)+b_9^{alg}\|\e^{\prime}\|_2. \label{eq:t23}
\end{eqnarray}
Using Lemma 2 to \eqref{eq:t21} and applying \eqref{eq:t22} and \eqref{eq:t23} to the resulting inequality, after some transforms, one has,
\begin{eqnarray}
\|\x_S-\x(k+1)\|_2&\le&(c_1^{alg}(\lambda_1^{alg})^{k-2}+c_2^{alg}(\lambda_2^{alg})^{k-2}+c_3^{alg} b^k)\|\x_S-\x(0)\|_2\nonumber\\
&&+(c_4^{alg}+c_5^{alg}(\lambda_1^{alg})^{k-2}+c_6^{alg}(\lambda_2^{alg})^{k-2}+c_7^{alg}b^k)\|\e^{\prime}\|_2.
\end{eqnarray}
where
\begin{eqnarray*}
c_1^{alg}&=&\left(b_1^{alg}b_8^{alg}+b b_5^{alg}\right)\left(\omega_1^{alg}+\frac{b_1^{alg}}{2}\right)+b_8^{alg}\left(\omega_2^{alg}+\frac{b_2^{alg}}{2}\right)+b^2b_5^{alg}\left(\omega_1^{alg}+\frac{b_1^{alg}}{2}\right)\theta_{11}^{alg},\\
c_2^{alg}&=&-\left(b_1^{alg}b_8^{alg}+b b_5^{alg}\right)\left(\omega_1^{alg}-\frac{b_1^{alg}}{2}\right)-b_8^{alg}\left(\omega_2^{alg}-\frac{b_2^{alg}}{2}\right)-b^2b_5^{alg}\left(\omega_1^{alg}-\frac{b_1^{alg}}{2}\right)\theta_{21}^{alg},\\
c_3^{alg}&=&\left(1+\left(\omega_1^{alg}+\frac{b_1^{alg}}{2}\right)\theta_{12}^{alg}-\left(\omega_1^{alg}-\frac{b_1^{alg}}{2}\right)\theta_{22}^{alg}\right)b_5^{alg},\\
c_4^{alg}&=&\left(1+\frac{\left(\omega_1^{alg}+\frac{b_1^{alg}}{2}\right)}{1-\lambda_1^{alg}}-\frac{\omega_1^{alg}-\frac{b_1^{alg}}{2}}{1-\lambda_2^{alg}}\right)b_7^{alg},\\
c_5^{alg}&=&\left(b_1^{alg}b_9^{alg}-b b_5^{alg}+b_7^{alg}\right)\left(\omega_1^{alg}+\frac{b_1^{alg}}{2}\right)+b_9^{alg}\left(\omega_2^{alg}+\frac{b_2^{alg}}{2}\right)\\
&&-b^2b_6^{alg}\left(\omega_1^{alg}+\frac{b_1^{alg}}{2}\right)-\frac{b_7^{alg}\left(\omega_1^{alg}+\frac{b_1^{alg}}{2}\right)}{1-\lambda_1^{alg}},\\
c_6^{alg}&=&-\left(b_1^{alg}b_9^{alg}-b b_5^{alg}+b_7^{alg}\right)\left(\omega_1^{alg}-\frac{b_1^{alg}}{2}\right)-b_9^{alg}\left(\omega_2^{alg}-\frac{b_2^{alg}}{2}\right)\\
&&+b^2b_6^{alg}\left(\omega_1^{alg}-\frac{b_1^{alg}}{2}\right)\theta_{21}^{alg}+\frac{b_7^{alg}\left(\omega_1-\frac{b_1^{alg}}{2}\right)}{1-\lambda_1^{alg}},\\
c_7^{alg}&=&-\left(1+\left(\omega_1^{alg}+\frac{b_1^{alg}}{2}\right)\theta_{12}^{alg}-\left(\omega_1^{alg}-\frac{b_1^{alg}}{2}\right)\theta_{22}^{alg}\right)b_6^{alg}.
\end{eqnarray*}

%\begin{eqnarray*}
%&&\rho^{IAD}=\frac{2(1-\sqrt{3}(|1-\mu|+\mu\delta_{3s}))}{\sqrt{3}\mu(1+\delta_{3s})},
%b_1^{IAD}=\sqrt{3}(|1-\mu|+\mu\delta_{3s}), \quad b_2^{IAD}=\frac{\sqrt{3}\mu|1-\gamma|(1+\delta_{3s})}{2};\\
%&&
%\rho^{NIAD}=\frac{2(1-(2\sqrt{3}+1)\delta_{3s})}{\sqrt{3}(1-\delta_{3s})^2(1+\delta_{3s})},
%b_1^{NIAD}=\frac{2\sqrt{3}\delta_{3s}}{1-\delta_{3s}}, \quad b_2^{NIAD}=\frac{\sqrt{3}|1-\gamma|(1+\delta_{3s})}{2(1-\delta_{3s})};\\
%&&
%\rho^{ADP}=\frac{\sqrt{2}(\sqrt{1-\delta_{3s}^2}-\sqrt{2}\delta_{3s})}{1+\delta_{3s}},
%b_1^{ADP}=\frac{2\delta_{3s}^2}{1-\delta_{3s}^2}, \quad b_2^{ADP}=|1-\gamma|\sqrt{\frac{1+\delta_{3s}}{2(1-\delta_{3s})}};
%\end{eqnarray*}
\be
b=\frac{1}{1+\gamma}<1,\nonumber
\ee
\begin{eqnarray*}
\lambda_1^{alg}=\frac{(b+b_1^{alg})+\sqrt{(b-b_1^{alg})^2+4 bb_2^{alg}}}{2}<1,  \lambda_2^{alg}=\frac{(b+b_1^{alg})-\sqrt{(b-b_1^{alg})^2+4bb_2^{alg}}}{2}<1.
\end{eqnarray*}

\be
\omega_1^{alg}=\frac{-b b_1^{alg}+(b_1^{alg})^2+2bb_2^{alg} }{2\sqrt{(b-b_1^{alg})^2+4 b b_2^{alg}}}, \quad
\omega_2^{alg}=\frac{(b+b_1^{alg})b_2^{alg}}{2\sqrt{(b-b_1^{alg})^2+4 b b_2^{alg}}}.
\ee
\begin{eqnarray*}
\theta_{11}^{alg}=\left\{
\begin{array}{cc}
\frac{1}{\lambda_1^{alg}-b}, & \lambda_1^{alg}\neq b \\
\frac{k-2}{\lambda_1^{alg}}, & \lambda_1^{alg}=b.
\end{array}\right.
 &,&
 \theta_{11}^{alg}=\left\{
\begin{array}{cc}
-\frac{1}{\lambda_1^{alg}-b}, & \lambda_1^{alg}\neq b \\
0, & \lambda_1^{alg}=b.
\end{array}\right. ;\\
\theta_{21}^{alg}=\left\{
\begin{array}{cc}
\frac{1}{\lambda_2^{alg}-b}, & \lambda_2^{alg}\neq b \\
\frac{k-2}{\lambda_2^{alg}}, & \lambda_2^{alg}=b.
\end{array}\right.
 &,&
 \theta_{22}^{alg}=\left\{
\begin{array}{cc}
-\frac{1}{\lambda_2^{alg}-b}, & \lambda_2^{alg}\neq b \\
0, & \lambda_2^{alg}=b.
\end{array}\right. .\\
\end{eqnarray*}
$\|\x_S-\x(k)\|_2$ will be bounded if $(1-b_1^{alg})(1-b)>b b_2$, which is equivalent to
\[
\rho^{alg}>\frac{|1-\gamma|}{\gamma},\]
where
\begin{eqnarray*}
&&\rho^{IAD}=\frac{2(1-\sqrt{3}(|1-\mu|+\mu\delta_{3s}))}{\sqrt{3}\mu(1+\delta_{3s})},\\
&&\rho^{NIAD}=\frac{2(1-(2\sqrt{3}+1)\delta_{3s})}{\sqrt{3}(1-\delta_{3s})^2(1+\delta_{3s})},\\
&&\rho^{ADP}=\frac{\sqrt{2}(\sqrt{1-\delta_{3s}^2}-\sqrt{2}\delta_{3s})}{1+\delta_{3s}}.
\end{eqnarray*}

Therefore, Theorem 1 is proved.

\subsection{Proofs of Corollary 1}
When $\x$ is $s$-sparse with support set $S$ and there is no noise,  if $\|\x(k+1)-\x\|_2<x_{\min}$, then $S(k+1)=S$. In this case, one can get the
exact solution after a posteriori least squares fit on $S(k+1)$. This case can be obtained if
\begin{eqnarray}
\|\x-\x(k+1)\|_2&<&(c_1^{alg}(\lambda_1^{alg})^{k-2}+c_2^{alg}(\lambda_2^{alg})^{k-2}+c_3^{alg}b^k)\|\x-\x(0)\|_2 \nonumber \\
&\le&(c_1^{alg}+c_2^{alg}+c_3^{alg}b^2)(\lambda^{alg})^{k-2}\|\x-\x(0)\|_2<x_{\min}. \label{eq:t31}
\end{eqnarray}
From $(c_1^{alg}+c_2^{alg}+c_3^{alg}b^2)(\lambda^{alg})^{k-2}\|\x-\x(0)\|_2<x_{\min}$ in \eqref{eq:t31}, one has
\[
k\ge\frac{\ln (x_{\min}/\|\x-\x(0)\|_2)-\ln(c_1^{alg}+c_2^{alg}+c_3^{alg}b^2)}{\ln \lambda^{alg}}+3.\]
Therefore, Corollary 1 is proved.
%\begin{eqnarray*}
%\|\x_S-\x(2)\|_2&\le&(3\delta_{3s}^2+\frac{\sqrt{3}\gamma(1+\delta_{3s})}{1+\gamma})\|\x_S\|_2\\
%&&+\sqrt{3(1+\delta_{3s})}(\sqrt{3}\delta_{3s}+1-\frac{\gamma}{2(1+\gamma)})\|\e^{\prime}\|_2
%\end{eqnarray*}
%
%
%\begin{eqnarray*}
%c(k)=b_1,d(k)=b_2,e(k)=b_3,l(k)=b_4\\
%b_1=\sqrt{3}\delta_{3s},b_2=\frac{\sqrt{3}|1-\gamma|(1+\delta_{3s})}{2},\\
%b_3=\sqrt{3}\gamma(1+\delta_{3s})\|\x\|_2-\frac{3}{1+\gamma}\|\e^{\prime}\|_2,b_4=\|\e^{\prime}\|_2.
%\end{eqnarray*}

%ADP:
%\begin{eqnarray}
%b_1=\sqrt{\frac{2\delta_{3s}^2}{1-\delta_{3s}^2}},b_2=\frac{|1-\gamma|(1+\delta_{3s})}{\sqrt{2(1-\delta_{3s}^2)}},\\
%b_3=\gamma\sqrt{\frac{2(1+\delta_{3s})}{1-\delta_{3s}}}\|\x_S\|_2-\frac{1}{1+\gamma}\sqrt{\frac{2}{1-\delta_{3s}}}\|\e^{\prime}\|_2,
%b_4=\left(\frac{\sqrt{1+\delta_{3s}}}{1-\delta_{3s}}+\sqrt{\frac{2}{1-\delta_{3s}}}\right)\|\e^{\prime}\|_2.
%\end{eqnarray}
%NIAD:
%\begin{eqnarray}
%b_1=\frac{2\sqrt{3}\delta_{3s}}{1-\delta_{3s}},b_2=\frac{\sqrt{3}|1-\gamma|(1+\delta_{3s})}{2(1-\delta_{3s})},\\
%b_3=\frac{\sqrt{3}\gamma(1+\delta_{3s})}{1-\delta_{3s}}\|\x_S\|_2-\frac{\sqrt{3(1+\delta_{3s}}}{(1+\gamma)(1-\delta_{3s})}\|\e^{\prime}\|_2,
%b_4=\frac{\sqrt{3(1+\delta_{3s})}}{1-\delta_{3s}}\|\e^{\prime}\|_2.
%\end{eqnarray}
\end{enumerate}
\bibliographystyle{IEEEtran}
\bibliography{Bib2}

\begin{thebibliography}{10}
\providecommand{\url}[1]{#1}
\csname url@samestyle\endcsname
\providecommand{\newblock}{\relax}
\providecommand{\bibinfo}[2]{#2}
\providecommand{\BIBentrySTDinterwordspacing}{\spaceskip=0pt\relax}
\providecommand{\BIBentryALTinterwordstretchfactor}{4}
\providecommand{\BIBentryALTinterwordspacing}{\spaceskip=\fontdimen2\font plus
\BIBentryALTinterwordstretchfactor\fontdimen3\font minus
  \fontdimen4\font\relax}
\providecommand{\BIBforeignlanguage}[2]{{%
\expandafter\ifx\csname l@#1\endcsname\relax
\typeout{** WARNING: IEEEtran.bst: No hyphenation pattern has been}%
\typeout{** loaded for the language `#1'. Using the pattern for}%
\typeout{** the default language instead.}%
\else
\language=\csname l@#1\endcsname
\fi
#2}}
\providecommand{\BIBdecl}{\relax}
\BIBdecl

\bibitem{donoho2006compressed}
D.~L. Donoho, ``Compressed sensing,'' \emph{IEEE Transactions on Information
  Theory}, vol.~52, no.~4, pp. 1289--1306, 2006.

\bibitem{candes2005decoding}
E.~J. Cand{\`e}s and T.~Tao, ``Decoding by linear programming,'' \emph{IEEE
  Transactions on Information Theory}, vol.~51, no.~12, pp. 4203--4215, 2005.

\bibitem{candes2006robust}
E.~J. Cand{\`e}s, J.~Romberg, and T.~Tao, ``Robust uncertainty principles:
  Exact signal reconstruction from highly incomplete frequency information,''
  \emph{IEEE Transactions on Information Theory}, vol.~52, no.~2, pp. 489--509,
  2006.

\bibitem{nesterov1994interior}
Y.~Nesterov, A.~Nemirovskii, and Y.~Ye, \emph{Interior-point polynomial
  algorithms in convex programming}.\hskip 1em plus 0.5em minus 0.4em\relax
  Philadelphia, PA: SIAM, 1994.

\bibitem{chambolle1998nonlinear}
A.~Chambolle, R.~A. De~Vore, N.-Y. Lee, and B.~J. Lucier, ``Nonlinear wavelet
  image processing: variational problems, compression, and noise removal
  through wavelet shrinkage,'' \emph{IEEE Transactions on Image Processing},
  vol.~7, no.~3, pp. 319--335, 1998.

\bibitem{beck2009fast}
A.~Beck and M.~Teboulle, ``A fast iterative shrinkage-thresholding algorithm
  for linear inverse problems,'' \emph{SIAM Journal on Imaging Sciences},
  vol.~2, no.~1, pp. 183--202, 2009.

\bibitem{figueiredo2007gradient}
M.~A. Figueiredo, R.~D. Nowak, and S.~J. Wright, ``Gradient projection for
  sparse reconstruction: Application to compressed sensing and other inverse
  problems,'' \emph{IEEE Journal of Selected Topics in Signal Processing},
  vol.~1, no.~4, pp. 586--597, 2007.

\bibitem{hale2008fixed}
E.~T. Hale, W.~Yin, and Y.~Zhang, ``Fixed-point continuation for
  $\ell_1$-minimization: Methodology and convergence,'' \emph{SIAM Journal on
  Optimization}, vol.~19, no.~3, pp. 1107--1130, 2008.

\bibitem{yang2011alternating}
J.~Yang and Y.~Zhang, ``Alternating direction algorithms for $\ell_1$-problems
  in compressive sensing,'' \emph{SIAM Journal on Scientific Computing},
  vol.~33, no.~1, pp. 250--278, 2011.

\bibitem{boyd2011distributed}
S.~Boyd, N.~Parikh, E.~Chu, B.~Peleato, and J.~Eckstein, ``Distributed
  optimization and statistical learning via the alternating direction method of
  multipliers,'' \emph{Foundations and Trends{\textregistered} in Machine
  Learning}, vol.~3, no.~1, pp. 1--122, 2011.

\bibitem{Blumensath2009}
T.~Blumensath and M.~E. Davies, ``Iterative hard thresholding for compressed
  sensing,'' \emph{Applied and Computational Harmonic Analysis}, vol.~27,
  no.~3, pp. 265--274, 2009.

\bibitem{blumensath2010normalized}
------, ``Normalized iterative hard thresholding: Guaranteed stability and
  performance,'' \emph{IEEE Journal of Selected Topics in Signal Processing},
  vol.~4, no.~2, pp. 298--309, 2010.

\bibitem{foucart2011hard}
S.~Foucart, ``Hard thresholding pursuit: an algorithm for compressive
  sensing,'' \emph{SIAM Journal on Numerical Analysis}, vol.~49, no.~6, pp.
  2543--2563, 2011.

\bibitem{maleki2010optimally}
A.~Maleki and D.~L. Donoho, ``Optimally tuned iterative reconstruction
  algorithms for compressed sensing,'' \emph{IEEE Journal of Selected Topics in
  Signal Processing}, vol.~4, no.~2, pp. 330--341, 2010.

\bibitem{wright2009sparse}
S.~J. Wright, R.~D. Nowak, and M.~A. Figueiredo, ``Sparse reconstruction by
  separable approximation,'' \emph{IEEE Transactions on Signal Processing},
  vol.~57, no.~7, pp. 2479--2493, 2009.

\bibitem{golub2012matrix}
G.~H. Golub and C.~F. Van~Loan, \emph{Matrix computations}.\hskip 1em plus
  0.5em minus 0.4em\relax JHU Press, 2012, vol.~3.

\bibitem{dai2009subspace}
W.~Dai and O.~Milenkovic, ``Subspace pursuit for compressive sensing signal
  reconstruction,'' \emph{IEEE Transactions on Information Theory}, vol.~55,
  no.~5, pp. 2230--2249, 2009.

\bibitem{song2013improved}
C.-B. Song, S.-T. Xia, and X.-j. Liu, ``{Improved Analyses for SP and CoSaMP
  Algorithms in Terms of Restricted Isometry Constants},'' \emph{arXiv preprint
  arXiv:1309.6073}, 2013.

\end{thebibliography}
\end{document}